\documentclass{article}
\usepackage[utf8]{inputenc}
\usepackage{authblk}
\usepackage{setspace}
\usepackage[margin=1.25in]{geometry}
\usepackage{graphicx}
\usepackage{subcaption}
\usepackage{amsmath}
\usepackage{amssymb}
\usepackage{amsthm}
\usepackage{bm}
\usepackage{hyperref}
\hypersetup{colorlinks=true,linkcolor=blue,citecolor=blue,urlcolor=blue}
\usepackage{xcolor}
\usepackage{booktabs}
\usepackage{float}
\usepackage{enumitem}
\setlist{itemsep=3pt,parsep=1pt}
\usepackage{tabularx}
\usepackage{array}
\usepackage{lineno}

\usepackage{microtype}
\usepackage{dblfloatfix}
\usepackage[section]{placeins}

\usepackage[style=phys,
citestyle=numeric-comp,
sorting=none,
doi=true,
url=false,
isbn=false,
natbib=true]{biblatex}
\newcolumntype{Y}{>{\raggedright\arraybackslash}X}

\newtheorem{theorem}{Theorem}
\newtheorem{proposition}{Proposition}
\newtheorem{definition}{Definition}
\newtheorem{corollary}{Corollary}
\newtheorem{remark}{Remark}
\newtheorem{lemma}{Lemma}

\title{A Minimal Four-Thruster System for Comet-Based Interstellar Navigation}

\author[1]{Bo Pieter Johannes Andr\'ee\thanks{The findings, interpretations, and conclusions expressed herein are entirely those of the author and do not necessarily represent the views of the International Bank for Reconstruction and Development/World Bank, its Board of Executive Directors, or the governments they represent. No World Bank resources were used to conduct this research. This research received no specific funding. Domain-specialist input is gratefully welcomed. Contact: bandree(at)worldbank.org}}

\affil[1]{Data Group and Multilateral \& UN, World Bank, Geneva, Switzerland.}
\date{March 2026}

\begin{document}

\maketitle

\begin{abstract}
Interstellar comets arrive with key ingredients for deep-space platforms already in place: volatile inventories convertible to propellant, natural rotation providing continuous attitude variation, and hyperbolic trajectories that carry them through the inner Solar System and back out to interstellar space. Rather than constructing spacecraft from scratch, we ask what \emph{minimal modification} is required to steer such a body along a controlled trajectory. The answer is surprisingly modest. By relaxing full six-degree-of-freedom control to forward-cone steering---sufficient for practical navigation---we show that \emph{four thrusters suffice}: one primary jet and three secondary jets at $120^\circ$ intervals. The secondary jets synthesize continuous in-plane steering, while the primary jet provides low-bandwidth attitude shaping: as the body rotates, the primary-jet torque direction sweeps predictably over a cycle, enabling out-of-plane steering via phase-scheduled firing. We formalize reachability under bounded-curvature constraints, characterize the rotation-mediated steering envelope, discuss enabling requirements including non-solar power at large heliocentric distances, and identify operational regimes and observable signatures implied by active trajectory control. The setting of a nutating axis is briefly considered and conjectured to preserve core results. The findings contribute to the broader effort of understanding the dynamics and control of small-body missions and offer a reference architecture relevant to long-horizon deep-space exploration and to potential planetary-defense concepts.
\end{abstract}

\pagebreak


\section{Introduction}
\label{sec:intro}

The detection of interstellar objects (ISOs) traversing the Solar System---most notably 1I/\hspace{0pt}'Oumuamua in 2017~\cite{meech2017oumuamua}, 2I/Borisov in 2019~\cite{jewitt2019comet}, and 3I/ATLAS in 2025~\cite{loeb2025intercepting}---has created unique opportunities to study material from other stellar systems and to establish an observational foundation for systematic characterization~\cite{eldadi2026ioss}. These discoveries further suggest that such bodies pass through the inner Solar System with sufficient regularity~\cite{delaFuenteMarcos2024ejected} to support mission planning, spurring proposals for intercept and rendezvous missions~\cite{hein2022interstellar,seligman2018oumuamua,hibberd2020project}. The growing catalog of ISOs and near-Earth small bodies alike has motivated both scientific exploration and planetary defense studies, with missions such as DART demonstrating the feasibility of altering asteroid trajectories through kinetic impact~\cite{rivkin2021dart}, and survey architectures being developed for monitoring potentially hazardous objects~\cite{zhou2022crown}.

Beyond their value as observational targets, ISOs also motivate mission architectures that exploit their kinematics for deep-space exploration. A more ambitious prospect is to treat suitable ISOs as \textit{platforms} for long-horizon navigation, not just flyby targets. Cometary bodies possess several characteristics advantageous for spacecraft design: natural rotation providing gyroscopic stability, volatile ices (H$_2$O, CO, CH$_4$, NH$_3$) constituting potential propellant~\cite{mumma2011composition,altwegg2019rosina}, organic compounds amenable to chemical processing~\cite{biver2019rosetta}, and trajectories that extend naturally to interstellar distances~\cite{jewitt2022interstellar}. The orbital dynamics and control challenges involved in operating near such bodies share common features with the broader class of small-body proximity operations that have been studied extensively in the context of asteroid rendezvous and exploration missions~\cite{lee2014hovering,qiao2023gravitational}.

This concept raises natural questions about the minimal propulsion architecture required to navigate such a platform, the operational modes it would employ, and how active control would alter observable signatures of the host body. We address these questions by shifting from instantaneous 6-DOF controllability to \emph{reachability} over practical mission timescales, where the dominant requirement is not active maneuvering but long-horizon \emph{trajectory tracking} under persistent disturbances~\cite{andree2026stability}---favoring dual architectures that enable slow retargeting and high-frequency stabilization. While the design primarily targets navigation, the concepts could also inform minimal-actuation architectures for planetary defense.

\subsection{Motivation: Comets as Interstellar Vehicles}

Traditional interstellar mission architectures face fundamental challenges: the rocket equation imposes exponential mass penalties for the enormous $\Delta V$ required, fuel must be transported from Earth at great expense, and achieving interstellar velocities from rest requires either decades of low-thrust acceleration or impractical energy expenditures. Alternative propulsion concepts, including solar sailing~\cite{macdonald2011solar}, have been explored for deep-space missions, but each faces distinct scalability constraints for interstellar distances.

Interstellar comets offer an alternative paradigm. These bodies arrive in the Solar System with hyperbolic excess velocities of 20--80~km/s~\cite{jewitt2022interstellar}---velocities that would require $>10^{15}$~J to achieve from Earth orbit for a modest payload. By \textit{attaching} a minimal propulsion system to an existing comet, one can leverage:

\begin{enumerate}
    \item \textbf{Pre-existing velocity}: The comet provides pre-existing initial acceleration equivalent to extreme gravity assists. Notably, gravity assist maneuvers are already standard practice in interplanetary mission design~\cite{wertz1999spacecraft,bolton2017jupiter}.
    
    \item \textbf{In-situ volatile inventory.} Cometary nuclei can contain substantial water and carbon-bearing volatiles that, in principle, could be processed into propellants~\cite{mumma2011composition}. While the extractable fraction and processing rate depend on energy availability, materials handling, and outgassing management, the accessible volatile reservoir on $\mathcal{O}(100\text{ m to }\text{several km})$ bodies is plausibly orders of magnitude larger than what is practical to transport from Earth for comparable deep-space maneuvering. In addition, a mission could exploit \emph{native outgassing} itself: modest engineering
    interventions (e.g., localized heating, venting, or surface conditioning) could bias naturally released gas into \emph{directed outgassing}. This would provide lower thrust authority than dedicated engines, but could reduce mechanical complexity and serve as a complementary trim or bias-thrust mode. This strengthens the case for minimal, robust control architectures optimized for longevity and reliability rather than aggressive $\Delta V$ performance.
    
    \item \textbf{Structural mass}: Large comets $\mathcal{O}(100\text{ m to }\text{several km})$ could support substantial payloads, including secondary spacecraft or communication arrays, without the exponential mass penalties of the rocket equation.
    
\item \textbf{Natural dynamics}: Cometary rotation periods typically span 5--70~hours, with the majority clustering between 5--20~hours~\cite{Kokotanekova2017Rotation}, consistent with observed values for interstellar comets~\cite{delaFuenteMarcos2025gtc,scarmato2026rotation}, providing gyroscopic stability that can be \textit{exploited} rather than counteracted for attitude control.

\end{enumerate}

This approach inverts the conventional spacecraft design problem. Instead of engineering a vehicle with full instantaneous maneuverability, we ask: \textit{what is the minimum thruster configuration required for three-dimensional navigation when forward motion and rotation are provided naturally?}

We establish that the answer is four thrusters: three small units arranged at $120^\circ$ intervals that can synthesize arbitrary \emph{transverse} thrust directions within a plane, plus one larger unit used primarily to manage the long-horizon \emph{orientation of that steering plane} through slow attitude adjustment. Natural rotation provides phase diversity and relaxes the requirement for instantaneous 6-DOF authority. The geometric mechanics underlying this approach---in particular, the exploitation of rotational symmetry to reduce control complexity---connects to a broader literature on attitude dynamics and control in strongly perturbed environments~\cite{sanyal2012attitude,zhou2024orbit}.

\paragraph{Scope of minimality.}
Throughout this paper, minimality is understood in a control-theoretic sense that encompasses
both hardware count and the dimensionality of the resulting control problem. Our starting
point is a body rotating steadily about a principal axis, which reduces the full attitude
state in~$SO(3)$ to a single predictable phase variable. This assumption is what enables the separation of fast in-plane steering from slow
out-of-plane reorientation that makes the four-thruster architecture operationally tractable
with minimal onboard autonomy. Relaxing this assumption---for instance, by admitting or
deliberately inducing non-principal-axis rotation---could yield alternative configurations
in which the functional roles of the thruster groups are redistributed, but the minimum
hardware count of four thrusters is preserved for the same geometric reasons
(Theorem~\ref{thm:120} governs torque synthesis no less than force synthesis). Such
alternatives would, however, involve a higher-dimensional coupled control problem and
correspondingly different operational and observational characteristics. We develop the
principal-axis architecture in full and consider the alternative design point in
Section~\ref{sec:nutation_architecture}.

\section{Theoretical Framework}
\label{sec:framework}

\subsection{Frames and kinematics}
\label{sec:frames_kinematics}

Because the nucleus rotates while the trajectory is defined in inertial space, we distinguish
a \emph{body-fixed} frame for thruster geometry from a \emph{trajectory} (navigation) frame for reachability.

\paragraph{Body frame $\mathcal B$ (geometry).}
Let $\mathcal B=\{\hat{\mathbf x}_B,\hat{\mathbf y}_B,\hat{\mathbf z}_B\}$ be a right-handed orthonormal frame
rigidly attached to the body, with $\hat{\mathbf z}_B$ the (approximately fixed) spin axis and $\hat{\mathbf x}_B$
a chosen body-fixed reference direction (e.g., a marked meridian or principal-axis reference), and
$\hat{\mathbf y}_B=\hat{\mathbf z}_B\times \hat{\mathbf x}_B$.
Thruster directions are fixed in $\mathcal B$:
\begin{itemize}
\item $\mathbf d_i^{(B)}\in\mathbb R^3$ denotes the \emph{unit thrust direction} of thruster $i$ in body coordinates,
      i.e.\ the direction of the applied force when thruster $i$ fires.
\item $u_i(t)\ge 0$ is the commanded \emph{thrust magnitude} (force) of thruster $i$ at time $t$
      (unidirectional jets; no thrust reversal).
\end{itemize}

\paragraph{Trajectory frame $\mathcal I$ (navigation).}
Let $\mathcal I=\{\hat{\mathbf x}_I(t),\hat{\mathbf y}_I(t),\hat{\mathbf z}_I(t)\}$ be a time-varying
right-handed orthonormal frame in inertial space. Equivalently, $\mathcal I(t)$ is represented by a rotation
matrix in $SO(3)$. Reachability and curvature statements are evaluated in $\mathcal I$.
A notational point deserves emphasis.

\begin{remark}[Space versus frame]
Both $\mathcal B$ and $\mathcal I$ are orthonormal frames \emph{in} $\mathbb R^3$ (not separate spaces).
They provide two coordinate representations of the same physical vectors, related by $R_{IB}(t)\in SO(3)$.
\end{remark}

A convenient choice is to align $\hat{\mathbf x}_I(t)$ with the instantaneous velocity direction,
\begin{equation}
\hat{\mathbf x}_I(t)=\frac{\mathbf v(t)}{\|\mathbf v(t)\|},
\end{equation}
where $\mathbf v(t)\in\mathbb R^3$ is the inertial velocity of the center of mass with $\|\mathbf v(t)\|>0$.
The remaining axes $\hat{\mathbf y}_I(t),\hat{\mathbf z}_I(t)$ can be any smooth
orthonormal completion (e.g., using an orbital normal when well-defined).

\paragraph{Rotation map and inertial force.}
Let $R_{IB}(t)\in SO(3)$ map body-frame vectors into inertial/trajectory coordinates. The thrust direction
of jet $i$ expressed in $\mathcal I$ is
\begin{equation}
\mathbf d_i^{(I)}(t)=R_{IB}(t)\,\mathbf d_i^{(B)},
\end{equation}
and the corresponding \emph{net thrust force} in inertial coordinates is
\begin{equation}
\mathbf F^{(I)}(t)=\sum_i u_i(t)\,\mathbf d_i^{(I)}(t),
\label{eq:inertial_force_def}
\end{equation}
where $\mathbf F^{(I)}(t)\in\mathbb R^3$ has units of force. Translational steering is governed by impulse:
\begin{equation}
\Delta \mathbf v=\frac{1}{m}\int_0^{t_f}\mathbf F^{(I)}(t)\,dt,
\label{eq:impulse_dv}
\end{equation}
with $m>0$ the (approximately constant) mass over the maneuver window and $\Delta \mathbf v$ the induced velocity change.
In free space, the trajectory responds directly to the net force $\mathbf F^{(I)}(t)$, while torques govern
attitude evolution. Torques matter for navigation here only through their effect on $R_{IB}(t)$, i.e.\ through how
they reorient the body-fixed thrust directions $\mathbf d_i^{(B)}$ in inertial space.

\paragraph{Steady-spin approximation.}
When the nucleus rotates near a principal axis with approximately constant rate $\omega_0$, define the spin phase
$\phi(t)=\omega_0 t$ and write
\begin{equation}
R_{IB}(t)\approx R_{IB}(0)\,R_{\hat{\mathbf z}_B}(\phi(t)),
\label{eq:R_decomp}
\end{equation}
where $R_{\hat{\mathbf z}_B}(\cdot)$ is a right-handed rotation about $\hat{\mathbf z}_B$. Under~\eqref{eq:R_decomp},
the inertial thrust directions $\mathbf d_i^{(I)}(t)$ and the net force $\mathbf F^{(I)}(t)$ in~\eqref{eq:inertial_force_def}
become explicit periodic functions of $\phi$, enabling closed-form identification of phases at which specific force/torque
components attain extrema.

\begin{remark}[Realism of the steady-spin assumption]
\label{rem:steady_spin_realism}
The principal-axis, constant-rate assumption is a simplification. Natural cometary nuclei are
often irregular bodies whose rotation states may include non-principal-axis components, and
even an initially well-behaved spin state may not persist over multi-decade mission durations
under stochastic outgassing torques, tidal interactions, and the cumulative effects of
thruster activity itself. We introduce this assumption to make the mathematical problem
tractable and to establish a minimal baseline architecture; the setting in which it is
relaxed is considered in Section~\ref{sec:nutation_architecture}.
\end{remark}

\subsection{Reachability: Formal Definition}
\label{sec:reachability}

We distinguish two complementary notions of thruster authority, evaluated in frames $\mathcal{B}$ and $\mathcal{I}$ respectively.

\emph{Instantaneous controllability} asks whether the thrusters can realize an arbitrary wrench (force and torque) at a given instant. This is the standard pointwise, algebraic feasibility question posed in the \emph{body frame} $\mathcal{B}$, where thrust directions and moment arms are fixed and the allocation matrix is defined---the setting that leads to conventional 6-DOF solutions. \emph{Reachability}, by contrast, asks whether the vehicle can achieve desired trajectory corrections---bounded heading and inclination changes---over finite time. This is a dynamical, integral property posed in the \emph{trajectory frame} $\mathcal{I}$.

The bridge between the two is the attitude map $R_{IB}(t) \in SO(3)$ (Section~\ref{sec:frames_kinematics}): body-fixed thrust directions $\mathbf{d}_i^{(B)}$ appear in inertial coordinates as $\mathbf{d}_i^{(I)}(t) = R_{IB}(t)\,\mathbf{d}_i^{(B)}$. Rotation therefore makes the available inertial thrust directions time-varying, so a configuration may fail instantaneous 6-DOF wrench controllability in $\mathcal{B}$ while still achieving practical 3D reachability in $\mathcal{I}$ by accumulating corrections through rotation-phased thrust histories.

\begin{definition}[Instantaneous controllability (wrench-level)]
\label{def:instant_control}
Fix a time $t$ and express wrenches in the body frame $\mathcal{B}$. Let $\mathbf{u}(t) = (u_1(t), \ldots, u_n(t))^\top \in \mathbb{R}_+^n$ collect the (unidirectional) thruster force magnitudes, and let $B \in \mathbb{R}^{6 \times n}$ denote the control-allocation matrix determined by the thruster directions and moment arms. The system is \emph{instantaneously controllable} if, for any desired wrench $\mathbf{w}^{(B)} = [(\mathbf{F}^{(B)})^\top, (\bm{\tau}^{(B)})^\top]^\top \in \mathbb{R}^6$, there exists $\mathbf{u}(t) \ge 0$ such that
\begin{equation}
B\,\mathbf{u}(t) = \mathbf{w}^{(B)}.
\end{equation}
\end{definition}

\begin{definition}[Practical forward-cone reachability (trajectory-level)]
\label{def:reachability}
A configuration provides \emph{practical forward-cone reachability} if, in the trajectory frame $\mathcal{I}$, the vehicle can steer its velocity direction within a bounded cone about the instantaneous forward axis $\hat{\mathbf{x}}_I(t)$ and execute sustained 3D trajectory corrections over mission-relevant timescales.

Let $\mathbf{F}^{(I)}(t)$ be the inertial net thrust force as in~\eqref{eq:inertial_force_def}, and let $\mathbf{a}^{(I)}(t) = \mathbf{F}^{(I)}(t)/m$ be the induced acceleration. Define the \emph{transverse component} (relative to the instantaneous velocity direction) by
\begin{equation}
a_\perp(t) \;\equiv\; \bigl\|(\mathbf{I} - \hat{\mathbf{x}}_I(t)\hat{\mathbf{x}}_I(t)^\top)\,\mathbf{a}^{(I)}(t)\bigr\|.
\end{equation}
Given a bound $a_\perp(t) \le a_\perp^{\max}$ and cruise speed $v_0 > 0$ (approximately constant over a maneuver window), the achievable curvature satisfies
\begin{equation}
\kappa_{\max} \approx \frac{a_\perp^{\max}}{v_0^2},
\qquad
\rho_{\min} \approx \frac{1}{\kappa_{\max}} = \frac{v_0^2}{a_\perp^{\max}},
\end{equation}
so sharper turns require either larger transverse acceleration or lower speed. We therefore seek configurations that (i)~realize inertial thrust directions spanning a forward cone about $\hat{\mathbf{x}}_I(t)$, and (ii)~can achieve both left/right and up/down steering over finite time, without requiring instantaneous 6-DOF wrench authority.
\end{definition}

Throughout this paper, ``3D reachability'' is understood in this practical sense: the ability to accumulate heading and inclination changes through admissible thrust histories, acknowledging that maneuver sharpness is governed by thrust-to-mass ratio, cruise speed, and the rotation-mediated availability of inertial thrust directions.

\section{The Four-Thruster Configuration}
\label{sec:config}

This section develops the minimal thruster configuration required for three-dimensional navigation of a naturally rotating body. We proceed in three steps: first establishing the minimum hardware for full in-plane control, then introducing the additional thruster needed for out-of-plane authority, and finally characterizing the resulting system's control envelope and operational constraints. 

\subsection{Operating Regimes: Cruise versus Encounter}
\label{sec:cruise_encounter_regimes}

Before deriving the configuration, we distinguish two operating regimes that impose different requirements on the propulsion system.

\textbf{Deep space cruise.} Over multi-decade trajectories, response delays of hours are operationally negligible. Small heading changes can be accumulated gradually, and modest coupling between attitude control and translation is tolerable because deviations can be corrected with low-cost trim maneuvers over long horizons. This regime favors hardware minimization and fuel efficiency over instantaneous responsiveness.

\textbf{Encounter phases.} During gravity assists, perihelion passages, or close flybys, the local dynamics evolve rapidly and navigation tolerances tighten. In these regimes, the main operational concern is not thrusting per se, but avoiding \emph{unintended net impulse} and plume--environment interactions (contamination, impingement, or poorly modeled coupling) at times when targeting and pointing margins are tight. The present architecture therefore admits two encounter-compatible modes depending on mission needs: (i) \emph{off/cold} operation with staged ramps when high responsiveness is not required, and (ii) \emph{symmetric warm standby} of the three small jets, which yields exactly zero net transverse force (Proposition~\ref{prop:warm_zero_force}) while maximizing immediate differential authority for disturbance rejection. Out-of-plane reorientation via the large thruster remains intrinsically low-bandwidth and is best handled by pre-phasing or by scheduling outside the most time-critical portions of an encounter.

This dichotomy is intrinsic to mission structure: scientific objectives are defined by encounters---flybys, gravity assists, perihelion passages---where dynamics intensify and tolerances tighten. For interstellar applications, hardware and latency constraints favor minimal actuation, yet the same thruster set must support both gradual cruise corrections and tight encounter traversals. The configuration developed below addresses both regimes with a single minimal hardware set.

\subsection{In-Plane Control: The 120$^\circ$ Placement Theorem}
\label{sec:120_theorem}

We first address the simpler problem of steering within a fixed plane---the transverse plane perpendicular to the spin axis. This corresponds to control authority in the body-fixed frame $\mathcal{B}$, where thruster directions remain constant. The fundamental constraint is that practical thrusters produce force in only one direction: they can push but not pull. This \emph{unidirectional thrust constraint} ($u_i \geq 0$) fundamentally limits the achievable force directions for any given thruster arrangement.

The natural question is: what is the minimum number of thrusters required to synthesize thrust in any direction within the plane? The answer is three, and optimal placement is at $120^\circ$ intervals.

\begin{theorem}[Minimum thrusters for full planar force-direction coverage]
\label{thm:120}
Let $\{\mathbf d_i\}_{i=1}^n\subset\mathbb R^2$ be unit thrust directions and let $u_i\ge 0$
(unidirectional thrust). Define the achievable set
\[
\mathcal C \equiv \left\{\sum_{i=1}^n u_i\mathbf d_i : u_i\ge 0\right\},
\]
i.e., the conic hull of the directions. Then:
\begin{enumerate}
\item[\textup{(i)}] If $n=2$, $\mathcal C$ is a wedge (a convex cone spanning at most $180^\circ$),
so $\mathcal C\neq\mathbb R^2$ and full directional coverage is impossible.
\item[\textup{(ii)}] For $n=3$, $\mathcal C=\mathbb R^2$ if and only if the three rays are not
contained in any closed half-plane (equivalently: no angular gap between consecutive directions is $\ge 180^\circ$).
\item[\textup{(iii)}] Among three-thruster configurations that satisfy $\mathcal C=\mathbb R^2$,
equal spacing at $120^\circ$ maximizes the minimum angular margin to the nearest thruster direction
(and yields the most symmetric worst-case directional authority).
\end{enumerate}
\end{theorem}

\begin{proof}
For $n=2$, any nonnegative combination lies in the cone spanned by the two rays, whose angular span is at most $180^\circ$.
Thus directions outside that wedge are unreachable.

For $n=3$, if the rays are contained in a closed half-plane, then all nonnegative combinations remain in that half-plane,
so $\mathcal C\neq\mathbb R^2$. Conversely, if they are not contained in any closed half-plane, then the origin lies in the
interior of the convex hull of $\{\mathbf d_i\}$, which implies $\mathcal C=\mathbb R^2$ (full directional coverage).

Finally, among feasible triples, the worst-case directional margin is governed by the largest angular gap between neighboring rays.
This worst-case is minimized by equal spacing, i.e., $120^\circ$ separation.
\end{proof}

In the notation of Section~\ref{sec:reachability}, this provides instantaneous controllability within the $(\hat{\mathbf{y}}, \hat{\mathbf{z}})$-plane of $\mathcal{B}$---but not yet the rotation-mediated authority in $\mathcal{I}$ required for practical forward-cone reachability.

\paragraph{Optimality.}

Two natural questions arise from Theorem~\ref{thm:120}: can fewer than three thrusters suffice, and does adding a fourth meaningfully help? Propositions~\ref{prop:insuf} and~\ref{prop:not4} address each in turn, and Figure~\ref{fig:coverage_conic_hull} visualizes the results geometrically.

\begin{proposition}[Insufficiency of two thrusters]
\label{prop:insuf}
Two thrusters, regardless of placement, cannot achieve bidirectional control in a 2D plane under unidirectional thrust constraints.
\end{proposition}

This follows immediately from Theorem~\ref{thm:120}(i): two rays span at most a $180^\circ$ wedge, leaving the opposite half-plane unreachable.

\begin{proposition}[Minimality of three thrusters]
\label{prop:not4}
Under the unidirectional thrust constraint ($u_i \geq 0$), three thrusters are the \emph{minimum} required to synthesize arbitrary planar thrust directions.
\end{proposition}

\begin{figure}[h!]
    \centering
    \includegraphics[width=\textwidth]{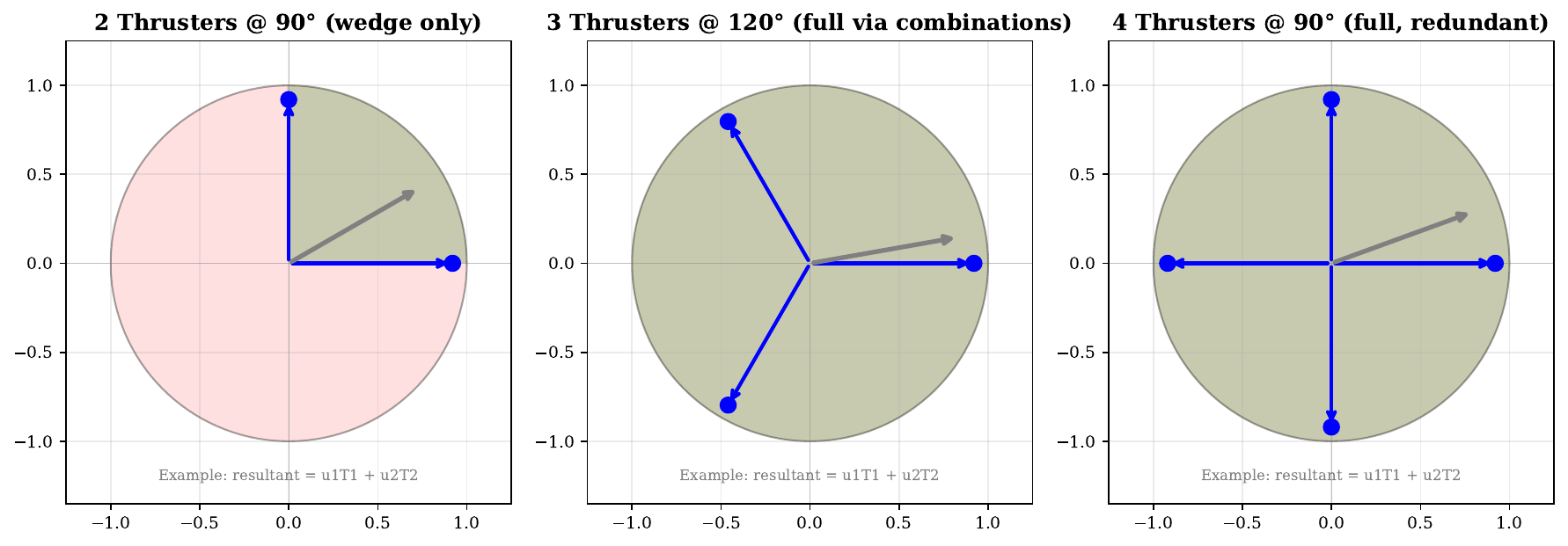}
    \caption{\textbf{Directional force coverage from nonnegative jet combinations.} 
    Blue arrows: individual jet directions; gray arrow: example resultant. 
    \textbf{Green}: reachable directions; \textbf{red}: unreachable. 
    Two jets at $90^\circ$ yield a wedge (left). 
    Three jets at $120^\circ$ achieve full planar coverage via pairwise combinations (middle). 
    Four jets at $90^\circ$ also yield full coverage but are redundant relative to the $120^\circ$ triad (right).}
    \label{fig:coverage_conic_hull}
\end{figure}

Deviation from equal spacing incurs quantifiable efficiency penalties. To see this, define \textit{worst-case control authority} as the minimum, over all target directions, of the ratio of achieved force magnitude to total thrust expended. For the optimal $120^\circ$ configuration, the worst-case direction lies at $60^\circ$ from any thruster (the midpoint of each gap), where achieving unit force requires exactly 2 units of total thrust, yielding worst-case authority of 0.5. For a misaligned configuration with angular gaps of $(100^\circ, 110^\circ, 150^\circ)$, numerical analysis shows the worst-case authority drops to 0.26---a 48\% reduction that nearly doubles propellant expenditure for maneuvers in unfavorable directions.

Additional thrusters (e.g., four at $90^\circ$ spacing) improve worst-case authority from $\cos 60^\circ = 0.5$ to $\cos 45^\circ \approx 0.71$, but increase hardware count and failure surface area without changing the geometric structure. The $120^\circ$ triad carries the further practical advantage of identical hardware, simplifying manufacturing and enabling interchangeability. Any deviation from equal spacing widens the largest angular gap, degrading efficiency at the midpoint; in the limit, a $180^\circ$ gap yields zero authority at $90^\circ$.

\subsection{Out-of-Plane Control: The Large Thruster Requirement}
\label{sec:large_thruster}

The three small thrusters at $120^\circ$ intervals provide full authority within their shared plane in $\mathcal{B}$, but cannot produce force components perpendicular to it. For three-dimensional navigation---that is, for practical forward-cone reachability in the inertial frame $\mathcal{I}$ as defined in Definition~\ref{def:reachability}---we require additional control authority in the out-of-plane direction.

Rather than adding three more small thrusters to create an independent perpendicular triad (which would double the hardware), we exploit the body's natural rotation. A single thruster mounted off the spin axis generates a moment arm that rotates with the body. While the thrust \emph{force} remains fixed in body coordinates, the resulting \emph{torque} sweeps through a full circle in inertial space over each rotation period. This enables attitude adjustments---and hence reorientation of the in-plane steering geometry---using a single additional thruster.

The required thrust magnitude for this fourth thruster depends on the body's rotational inertia and the desired reorientation rate.

\begin{theorem}[Torque requirement for spin-axis reorientation]
\label{thm:torque_req}
For a thruster exploiting body rotation to provide phase-dependent torque authority, the required thrust magnitude must overcome rotational inertia:
\begin{equation}
    F_{\text{large}} \geq \frac{I \dot{\omega}_{\text{req}}}{R}
\end{equation}
where $I$ is the moment of inertia, $\dot{\omega}_{\text{req}}$ is the required angular acceleration, and $R$ is the moment arm.
\end{theorem}

The physical interpretation is that the large thruster must generate sufficient torque to change the body's attitude against its rotational inertia. For a spinning body, the effect of torque depends on its orientation relative to the angular momentum $\mathbf{L}=I\boldsymbol{\omega}$: torque components parallel to $\mathbf{L}$ primarily change spin rate, whereas transverse components induce precession and reorientation of the spin axis. In the present architecture, the three small thrusters can support frequent, low-impulse trim actions (including spin-rate management), while the large thruster provides the slower transverse torque authority needed to reorient the steering geometry over long horizons.

Concretely, a single large thruster at position $(0, R, 0)$ firing in direction $(0, 0, 1)$ generates torque that, in the inertial frame, traces:
\begin{equation}
    \bm{\tau}(t) = F_{\text{large}} R \begin{pmatrix} \cos(\omega_0 t) \\ -\sin(\omega_0 t) \\ 0 \end{pmatrix}
\end{equation}
Over one rotation period, this sweeps the entire $x$-$y$ \emph{torque} plane, providing phase-dependent attitude authority. Note carefully the distinction: while the thrust \emph{force} direction remains fixed in body coordinates (along the spin axis), the \emph{torque} direction rotates in inertial space because the moment arm sweeps through orientations. This is the key mechanism enabling out-of-plane control with a single additional thruster.

\paragraph{Forward-cone reachability.}
With the addition of the large thruster, we can now establish that the four-thruster configuration satisfies the practical reachability requirements of Definition~\ref{def:reachability}.

\begin{proposition}[Four thrusters suffice for forward-cone reachability]
\label{prop:four_thruster_reachability}
Consider a rotating body with spin axis $\hat{\mathbf{z}}$ and angular velocity $\omega_0 = 2\pi/T$, equipped with:
\begin{enumerate}
\item[\textup{(i)}] Three small thrusters at $120^\circ$ intervals in the $(\hat{\mathbf{x}}, \hat{\mathbf{y}})$-plane of $\mathcal{B}$, providing instantaneous in-plane force authority up to magnitude $F_{\mathrm{small}}$;
\item[\textup{(ii)}] One large thruster generating thrust $F_{\mathrm{large}}$ along $\hat{\mathbf{z}}$, with moment arm $R$ about the spin axis.
\end{enumerate}
Then the configuration provides practical forward-cone reachability in $\mathcal{I}$ (Definition~\ref{def:reachability}):
\begin{enumerate}
\item[\textup{(a)}] \emph{In-plane steering} is available instantaneously, with maximum transverse acceleration $a_\perp^{\max} = F_{\mathrm{small}}/m$;
\item[\textup{(b)}] \emph{Out-of-plane steering} is available with latency at most $T/4$, via rotation-phased firing that selects the torque direction in $\mathcal{I}$;
\item[\textup{(c)}] \emph{Arbitrary 3D heading changes} within the forward cone are achievable with latency at most $T/2$, through combined in-plane thrust and attitude reorientation.
\end{enumerate}
\end{proposition}

\begin{proof}
Part (a) follows from Theorem~\ref{thm:120}: the $120^\circ$ triad achieves $\mathcal{C} = \mathbb{R}^2$ in the steering plane, so any transverse direction is instantaneously accessible.

For part (b), the large thruster's torque vector sweeps through all directions in the transverse torque plane over period $T$. Any desired torque direction is therefore achieved at some phase $\phi^* \in [0, 2\pi)$, and the maximum wait time for that phase is $T/2$. However, since torque directions $\pm\hat{\mathbf{n}}$ produce opposite angular accelerations, the effective wait time is at most $T/4$.

Part (c) follows by composition: in-plane components are addressed instantaneously while out-of-plane components require at most $T/4$ for favorable phasing, giving a worst-case combined latency of $T/2$ for heading changes requiring both.
\end{proof}

\begin{figure}[h!]
    \centering
    \includegraphics[width=\textwidth]{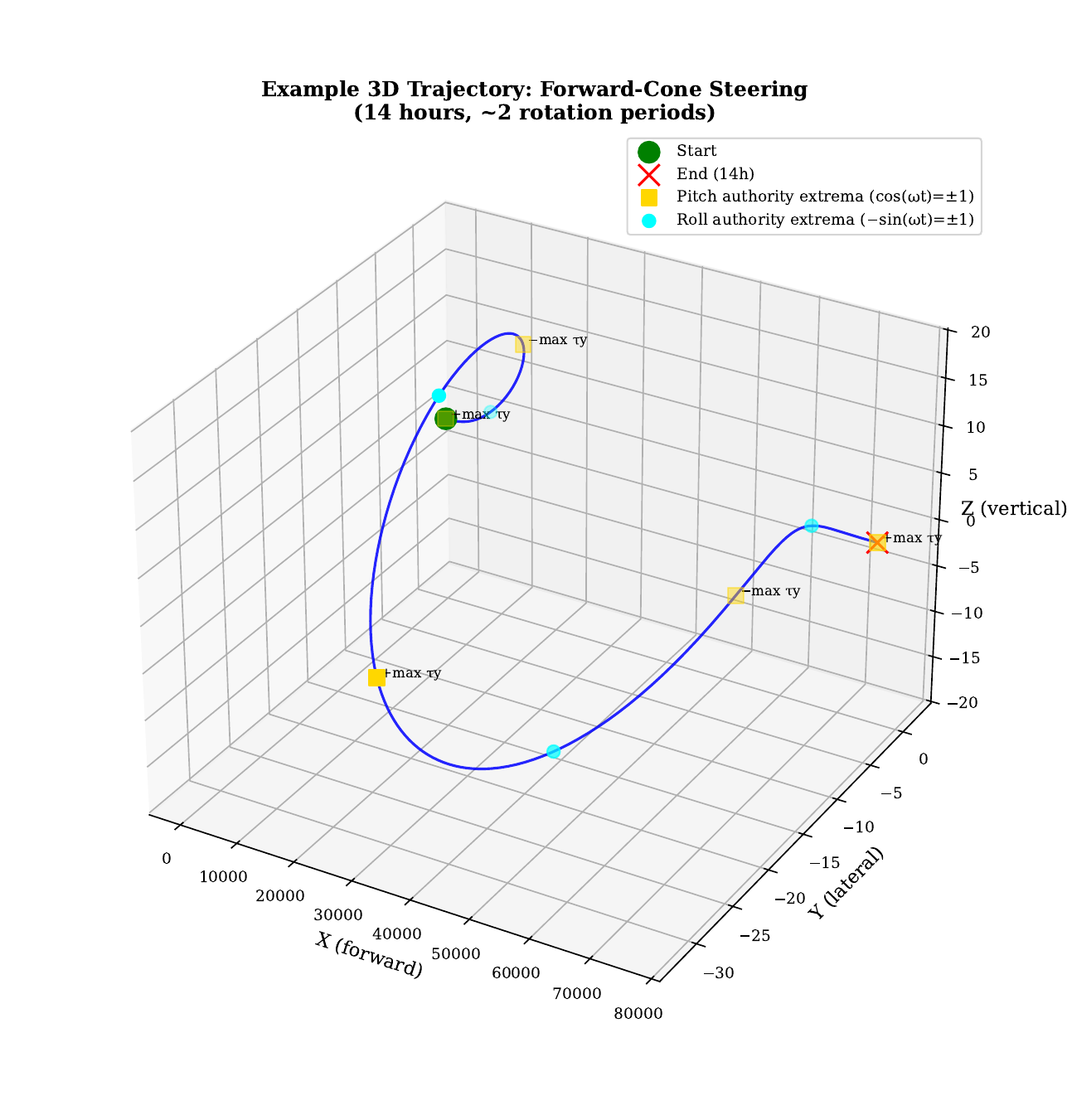}
    \caption{\textbf{Example forward-cone steering trajectory.} 
    A 3D path over 14 hours ($\approx$2 rotations) with steering constrained within a forward cone. 
    Markers indicate spin phases at which pitch-axis torque reaches extrema (squares; $\cos\omega t = \pm 1$) 
    and roll-axis torque reaches extrema (circles; $-\sin\omega t = \pm 1$), illustrating how 
    sign-controllable out-of-plane authority becomes available at predictable phases.}
    \label{fig:trajectory_cone_steering}
\end{figure}

Figure~\ref{fig:trajectory_cone_steering} illustrates the resulting capability: by firing the large thruster at appropriate rotation phases, the system executes controlled 3D trajectory modifications while the markers indicate the predictable phases at which specific torque directions become available.

\begin{remark}[Role of a second large thruster]
\label{rem:second_large}
A second large thruster is not required for the \emph{existence} of long-horizon 3D steering, but it can provide 
(i) redundancy, 
(ii) improved decoupling between attitude torque and net translational impulse by enabling partial cancellation of undesired drift during attitude adjustments, and 
(iii) reduced latency---with two thrusters at opposite positions, the maximum wait for a favorable torque phase is halved from $T/4$ to $T/8$.
\end{remark}

For missions where reliability margins justify additional hardware, a second large thruster (positioned symmetrically) allows attitude maneuvers with reduced translational side-effects. However, the minimal configuration requires only one.

\subsection{The Minimal Configuration}
\label{sec:minimal_config}

Combining the results above, we arrive at the minimal configuration for three-dimensional forward-cone steering: three small thrusters at $120^\circ$ intervals for in-plane control, plus one large thruster for rotation-mediated out-of-plane authority. This asymmetric allocation is not an arbitrary design choice but emerges from the structure of the control problem.

\begin{corollary}[Optimality of asymmetric allocation]
\label{cor:asymmetric}
For a rotating body with forward motion, the fuel-optimal configuration for 3D reachability consists of three small thrusters at $120^\circ$ for instantaneous planar control plus one large thruster for rotation-mediated perpendicular control.
\end{corollary}

The proof follows from the decomposition of the control problem into independent subspaces: in-plane steering (continuous, frequent) optimally uses small thrusters, while out-of-plane reorientation (intermittent, delay-tolerant) optimally uses one large thruster exploiting rotation. The frequency separation between these control modes justifies the asymmetric sizing: small thrusters handle high-bandwidth translational trim, while the large thruster handles low-bandwidth attitude shaping.

\begin{remark}[Minimality beyond hardware count]
\label{rem:dimensionality}
The minimality of the present architecture extends beyond thruster count to the dimensionality
of the control problem itself. Under the steady-spin
approximation~\eqref{eq:R_decomp}, the full attitude state in~$SO(3)$---parameterized by
nutation angle~$\theta$, precession angle~$\psi$, and spin phase~$\phi$---collapses to a
single periodic variable~$\phi(t)=\omega_0 t$. Attitude enters the translational control
problem only as a predictable clock for phase scheduling, yielding an effectively
one-dimensional attitude manifold. This dimensional reduction is what enables the separation
between instantaneous in-plane steering and low-bandwidth out-of-plane reorientation that
underlies Proposition~\ref{prop:four_thruster_reachability}. Relaxing the principal-axis
assumption would preserve the four-thruster minimum (Theorem~\ref{thm:120} applies equally
to torque synthesis) but would elevate the attitude control problem to its full
three-dimensional form, with nutation angle becoming an actively managed state variable
coupled to the translational steering envelope. We return to this point in
Section~\ref{sec:nutation_architecture}.
\end{remark}

Figure~\ref{fig:config} illustrates the configuration; Table~\ref{tab:config} provides specifications.

\begin{table}[h!]
\caption{Optimal four-thruster configuration for body radius $R$.}
\label{tab:config}
\centering
\small
\setlength{\tabcolsep}{4pt}
\renewcommand{\arraystretch}{1.05}
\begin{tabular}{@{}lllll@{}}
\toprule
ID & Position & Direction & Type & Function \\
\midrule
T1 & $(R,0,0)$ & Tangent ($0^\circ$)   & Small & In-plane steering \\
T2 & $120^\circ$ from T1 & Tangent & Small & In-plane steering \\
T3 & $240^\circ$ from T1 & Tangent & Small & In-plane steering \\
T4 & $(0,R,0)$ & $(0,0,1)$ & Large & Attitude shaping \\
\bottomrule
\end{tabular}
\end{table}

\begin{figure}[h!]
    \centering
    \includegraphics[width=\textwidth]{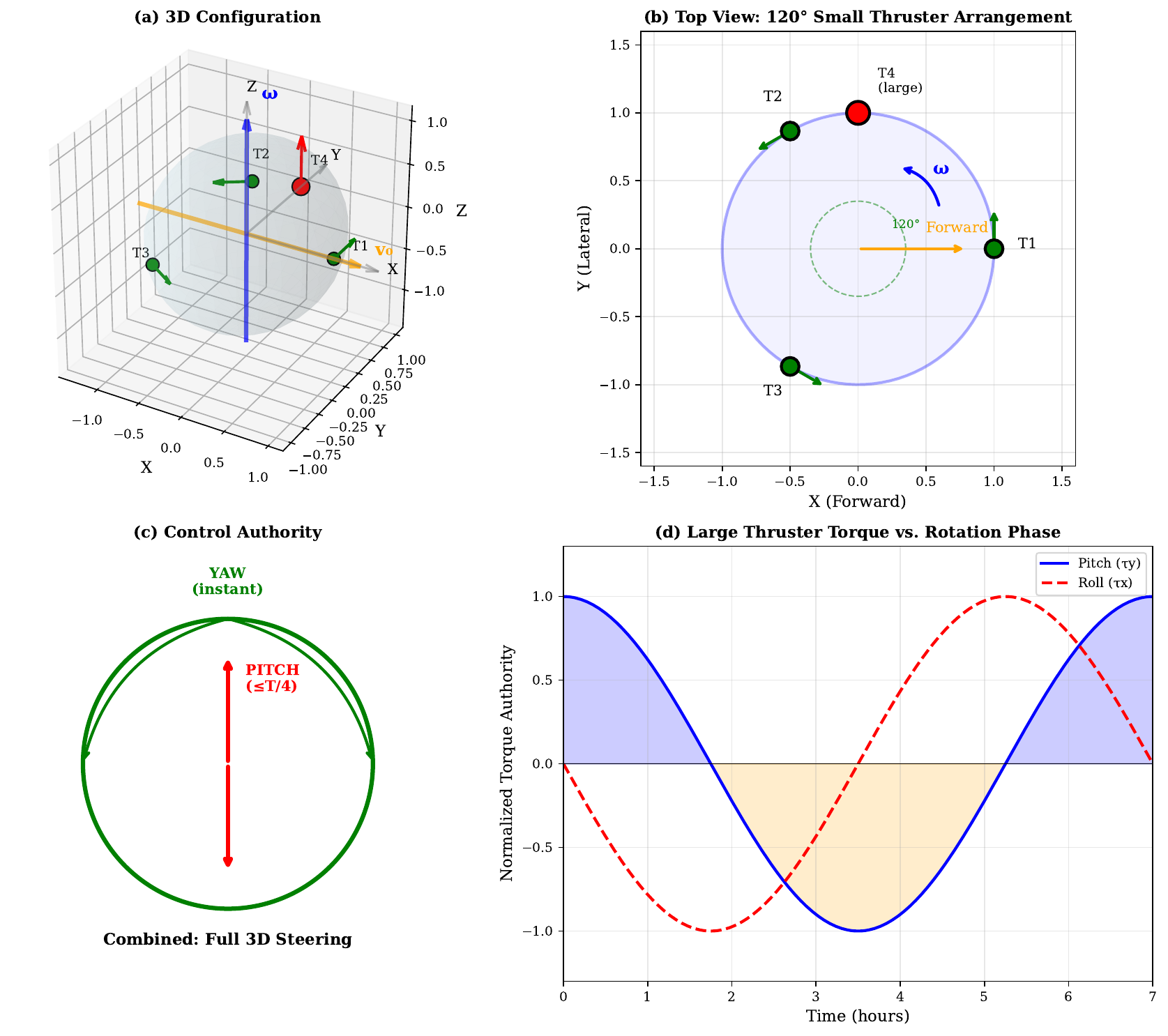}
    \caption{\textbf{Minimal four-thruster configuration for practical forward-cone steering.} \textbf{(a)} Three-dimensional view showing the spherical body, three small thrusters (green, T1--T3) at $120^\circ$ intervals in the equatorial plane used to synthesize transverse (in-plane) thrust vectors for heading control, one large thruster (red, T4) providing low-bandwidth attitude authority for long-horizon reorientation, the rotation axis $\omega$ (blue), and forward velocity $v_0$ (orange). \textbf{(b)} Top view emphasizing the $120^\circ$ spacing of the small thrusters, which enables bidirectional planar thrust synthesis under unidirectional constraints. \textbf{(c)} Conceptual control authority: instantaneous in-plane steering (green) and slower out-of-plane steering enabled by attitude shaping (red). \textbf{(d)} Large-thruster torque direction versus rotation phase, illustrating phase-dependent attitude authority that sweeps a transverse torque plane over one period.}
    \label{fig:config}
\end{figure}

\subsection{Response Times and Rotation Constraints}
\label{sec:response_rotation}

Theorems~\ref{thm:120}--\ref{thm:torque_req} and Proposition~\ref{prop:four_thruster_reachability} establish \emph{instantaneous} control authority---what forces and torques can be generated at any given moment. However, practical forward-cone reachability (Definition~\ref{def:reachability}) is a \emph{time-series} concept: it concerns trajectories achievable through thrust histories, not snapshots of available actuation. The spatial feedback geometry established above is necessary for trajectory control stability, but not sufficient on its own, since it does not prevent stochastic disturbances from accumulating into long-run trajectory divergence. A complete guarantee additionally requires that the closed-loop dynamics contract over time under the joint spatial-temporal structure~\cite{andree2020dynamic}. Establishing such contraction would require parameterizing the force equations over time, specifying a disturbance process, and verifying stability of the resulting dynamical system, which is beyond the present scope. Here we address only the minimal temporal consideration: the rotation period $T$, which governs how instantaneous authority in $\mathcal{B}$ maps to accumulated reachability in $\mathcal{I}$. Specifically, the rotation constraints explored below identify the regime under which the spatial authority results do not break down---ensuring that rotation-mediated control remains operationally viable. Together with the force and torque geometry of the preceding sections, this provides physically realisable force equations so that higher-level stability conditions rest on a well-posed foundation.

The effectiveness of rotation-mediated control depends on $T$ falling within an operationally viable range. Counter-intuitively, the limit $T \to 0$ does not converge to instantaneous static control---$\mathcal{I}$ and $\mathcal{B}$ remain distinct frames regardless of rotation rate---and the physical demands on the system diverge, making control progressively harder rather than easier. Specifically, three effects degrade controllability as $T$ decreases: the phase window for favorable torque direction narrows proportionally with $T$, making precise firing synchronization challenging for $T \lesssim 1$~hour; gyroscopic resistance scales as $\omega = 2\pi/T$, so required thruster force scales as $T^{-1}$; and centripetal acceleration at radius $R$ scales as $T^{-2}$, imposing structural constraints on thruster attachment. In the opposite limit, excessively slow rotation introduces its own penalties: maximum response time for out-of-plane maneuvers scales as $T/2$, so $T \gtrsim 100$~hours implies multi-day delays, and gyroscopic stability degrades as angular momentum $L \propto T^{-1}$ decreases, leaving the body susceptible to uncontrolled tumbling from stochastic disturbances.

These bounds suggest approximately $1 \lesssim T \lesssim 100$~hours, with performance heuristically favoring $T \sim 5$--$20$~hours---a range that coincides with observed cometary rotation periods~\cite{Kokotanekova2017Rotation}, suggesting that many potential target bodies may possess intrinsically suitable dynamics without requiring active spin modification. Bodies with $T$ outside this range would require either spin-up/spin-down maneuvers during the attachment phase or alternative control architectures with additional hardware. Table~\ref{tab:response} summarizes the resulting response times for a representative period of $T = 7$~hours: in-plane maneuvers are instantaneous, out-of-plane maneuvers require at most $T/4$ for favorable phasing, and arbitrary 3D heading changes require at most $T/2$. For interstellar cruise timescales, such delays are negligible.

\begin{table}[h!]
\caption{Response times for an example rotation period $T = 7$~h.}
\label{tab:response}
\begin{center}
\small
\begin{tabular}{@{}lll@{}}
\toprule
Maneuver & Time & Mechanism \\
\midrule
In-plane & Instant & Small thrusters \\
Out-of-plane & $\leq T/4$ & Large thruster + rotation \\
Arbitrary 3D & $\leq T/2$ & Combined \\
\bottomrule
\end{tabular}
\end{center}
\end{table}

\section{Application to Solar System Encounter}
\label{sec:application}

Having established the core control-theoretic results in the preceding sections, we now turn to operational interpretation: how encounter geometry, restart constraints, and observability shape feasible operating modes. The four-thruster configuration applies broadly to any naturally rotating body with forward motion; here we specialize to the motivating case of an interstellar object transiting a stellar system, where trajectory geometry and the presence of a dominant gravitational primary jointly constrain both feasible operations and preferred configurations. To keep the main narrative compact, we collect the remaining formal statements and proofs in two appendices---Appendix~\ref{sec:thruster_theory} for the three-jet system and Appendix~\ref{sec:axial_thrust_theory} for the axial jet---and reference them where used with visual guidance where helpful.

\subsection{Mission Architecture}
\label{sec:mission_arch}

A minimum viable architecture for comet-based interstellar exploration can be organized into five stages:

\begin{enumerate}
    \item \textbf{Detection}: Identify an incoming ISO via survey telescopes (e.g., the Vera C. Rubin Observatory)~\cite{engelhardt2017population}.
    
    \item \textbf{Intercept}: Launch a fast intercept mission to rendezvous with the ISO during its Solar System passage~\cite{landau2023iso,loeb2025intercepting}.
    
    \item \textbf{Attachment}: Deploy a four-thruster assembly onto the cometary surface.
    
    \item \textbf{Fuel processing}: Extract and process volatiles from cometary ices into propellant.
    
    \item \textbf{Transit and encounter}: Navigate through the target stellar system, executing reconnaissance flybys or gravity assists as mission objectives require.
\end{enumerate}

Table~\ref{tab:architecture} summarizes how this architecture shifts the trade space from launch $\Delta V$ and Earth-carried propellant toward longevity, robustness, and control with minimal added hardware.

\begin{table}[h!]
\caption{Comparison of mission architectures.}
\label{tab:architecture}
\centering
\small
\setlength{\tabcolsep}{4pt}
\renewcommand{\arraystretch}{1.05}
\begin{tabular}{@{}lll@{}}
\toprule
Property & Traditional & Comet-Based \\
\midrule
Launch $\Delta V$ & $>$50~km/s & $\sim$0 incremental \\
Fuel source & Earth-carried & In-situ processing \\
Payload limit & Rocket equation & Structural capacity \\
Thrusters needed & 6+ (full 6-DOF) & 4 (minimal) \\
\bottomrule
\end{tabular}
\end{table}

These potential benefits come with corresponding technical constraints. Rendezvous with a hyperbolic object requires rapid detection and response, though recent analyses suggest feasibility with existing or near-term propulsion~\cite{landau2023iso,loeb2025intercepting,linares2024statite}. Beyond $\sim$10~AU, solar power is often impractical for sustained high-power operation, favoring nuclear sources (RTGs or compact fission systems) that can sustain operation over multi-decade horizons~\cite{gibson2018krusty}. Securing thrusters to an irregular, potentially outgassing body presents significant engineering challenges, though surface operations on comets have been demonstrated by Philae/Rosetta~\cite{altwegg2019rosina}.

The most stringent requirements arise during \emph{encounters}---transits through a stellar system that involve close approaches to planets or the primary. Scientific flybys, gravity assists, and communication windows depend on meeting encounter geometries while maintaining hazard clearance margins. The remainder of this section develops how these constraints shape operating regimes and, in the near-primary phase, a preferred propulsion configuration.

\subsection{Mission Trajectory Geometry}
\label{sec:trajectory_geometry}

The four-thruster architecture derives additional justification from the geometric constraints governing stellar system transit. At ISO-class velocities (up to $\sim$80~km/s), the inner Solar System is crossed in weeks, and large lateral $\Delta V$ for close planetary flybys may be prohibitive. Trajectory geometry therefore constrains encounter distances, and slope selection becomes a first-order mission parameter.

\paragraph{Debris avoidance.}
Let the debris structure have characteristic heliocentric radius $r_{\rm belt}$ and effective vertical half-opening angle $i_{\rm eff}$, giving conservative effective half-thickness $z_{\rm eff}\approx r_{\rm belt}\tan i_{\rm eff}$. For an opposite-side crossing---entering above the debris midplane and exiting below it (or vice versa)---the minimum required inclination is
\begin{equation}
\theta^* \;=\; \arctan\!\left(\frac{2z_{\rm eff}}{d}\right),
\end{equation}
where $d$ is the in-plane chord length through the belt region. For Kuiper-belt-scale parameters ($r_{\rm belt}\sim 50~{\rm AU}$, $d\sim 100~{\rm AU}$, $i_{\rm eff}\sim 5^\circ$--$7^\circ$), this yields $\theta^*\sim 5^\circ$--$7^\circ$.

For inner-system targets, encounter range depends on both in-plane alignment and vertical offset. A trajectory crossing the reference plane near a priority target minimizes vertical offset for that encounter. For targets away from the crossing point, vertical offset depends on horizontal alignment, with practical benefit of slope reduction depending on the trajectory and planetary positions. A reduced slope may independently benefit alignment-sensitive measurements. Combining debris avoidance with these reconnaissance objectives suggests a slope \emph{profile}: higher slope ($\sim$7$^\circ$) during belt transit; a potentially reduced slope ($\sim$5$^\circ$) during inner-system passage; then restored higher slope for outbound crossing.

\paragraph{Design implications.}
This geometry suggests an asymmetry in propulsion requirements:
\begin{enumerate}
\item \textbf{Steering demands are comparatively low-bandwidth.} Degree-scale trajectory modifications (e.g., $2\times\sim$2$^\circ$ adjustments) can be executed over weeks to months. This is the natural domain of the large thruster (T4): slow attitude shaping that reorients the steering plane.

\item \textbf{Stabilization demands can be tighter.} Once a belt-crossing corridor and encounter geometry are selected, clearance margins must be maintained against accumulated perturbations. This is the natural domain of the small thrusters (T1--T3): continuous, precise trim and disturbance rejection.

\item \textbf{Trajectory replanning is costly in flight.} The coupled requirements of debris avoidance, encounter placement, and instrument geometry may not be readily re-optimized during transit, so operations often emphasize maintaining the planned corridor.
\end{enumerate}

\subsection{Operating Regimes}
\label{sec:operating_regimes}

The propulsion system supports a small set of operating regimes that are only weakly dependent on the local gravitational environment and instead reflect actuator constraints (restart reliability, minimum stable flow, ignition transients) and mission-phase objectives (persistent disturbance rejection versus discrete retargeting).

\paragraph{(R0) Off / pulsed operation.}
Thrusters remain off except for discrete maneuvers. This regime minimizes propellant expenditure in quiescent intervals, but incurs repeated ignition cycles and exposure to startup transients. In the three-jet geometry, ignition asynchrony produces a non-negligible impulse error that scales with commanded throttle (Proposition~\ref{prop:ignition_transient}), making high-throttle cold starts undesirable in precision-navigation windows. Moreover, cold-start synthesis can require thrust inflation relative to warm differential steering, with a worst-case thrust inflation of a factor of two for mid-gap force directions (Proposition~\ref{prop:cold_start_efficiency}).

\paragraph{(R1) Symmetric warm standby (small-jet triad).}
The three in-plane thrusters operate at a common baseline $u_1=u_2=u_3=u_0$. By Proposition~\ref{prop:warm_zero_force}, the net in-plane force vanishes exactly, while maintaining readiness for rapid differential actuation. This means that nonzero plume activity is present without inducing transverse acceleration---a decoupling that has consequences for observational interpretation (Section~\ref{sec:signatures}). Warm standby is operationally attractive when frequent restarts are costly or when persistent, low-level disturbance rejection is anticipated.

\paragraph{(R2) Warm differential steering (small-jet triad).}
Steering is implemented by reallocating thrust around a warm baseline, writing $u_i=u_0+\delta_i$, so the commanded in-plane force depends on the differentials (Proposition~\ref{prop:differential_decomposition}). Two sub-regimes are operationally relevant:
\begin{enumerate}
    \item \textbf{(R2a) Thrust-neutral reallocation:} $\sum_i \delta_i=0$ changes acceleration direction without changing total small-jet throughput (Proposition~\ref{prop:thrust_neutral_unique}). This is relevant when feed systems prefer approximately steady mass flow.
    \item \textbf{(R2b) Decrement-only steering:} for sufficiently high warm bias, steering may be executed using only bounded reductions relative to $u_0$ (Proposition~\ref{prop:decrement_only}), which is attractive when avoiding ramp-up transients is critical.
\end{enumerate}

Two bias choices are natural depending on the constraint of interest. Mid-throttle bias $u_0=u_{\max}/2$ maximizes symmetric headroom for both upward and downward differentials (Proposition~\ref{prop:headroom_mid}). High bias $u_0=u_{\max}$ maximizes decrement-only authority (Proposition~\ref{prop:decrement_only}) and is natural when the operational objective is to avoid ignition or ramp-up transients.

\paragraph{(R3) Retargeting (large jet).}
The axial thruster provides the primary authority for slow reorientation and coarse trajectory shaping, typically executed as intermittent maneuvers on timescales longer than the host rotation period.

\paragraph{(R4) Bias operation (large jet).}
When continuous operation is preferred (e.g., due to feed-system stability or persistent disturbance rejection), the axial thruster may be held at low duty cycle or low throttle to provide a steady bias acceleration. The preferred bias direction is environment-dependent; in the near-primary encounter regime, two-body dynamics single out the radial line as a convenient bias channel with reduced cross-track coupling; the corresponding formal statements are collected in Appendix~\ref{sec:axial_thrust_theory}.

Operationally, a natural sequencing is intermittent retargeting in (R3), followed by return to low-level bias operation (R4) once the trajectory converges, thereby reducing cross-track excitation in the near-primary regime (Section~\ref{sec:operational_config}).

\paragraph{Force--torque implications.}
In symmetric warm standby (R1), the small-jet triad cancels net in-plane \emph{force} exactly (Proposition~\ref{prop:warm_zero_force}), but generally does not cancel net \emph{torque} about the spin axis, implying slow secular changes in rotation unless corrected. Conversely, continuous axial-jet operation (R4) produces a nonzero translational impulse and therefore secular $\Delta V$ drift. Practically neutral ``standby'' behavior therefore requires either phase-aware firing schedules that reduce net impulse over a rotation cycle, or explicit opposing capability to decouple attitude control from translation. These considerations shape both feasible operations and the resulting signature profiles.

\subsection{Operational Configuration Near a Primary}
\label{sec:operational_config}

Section~\ref{sec:operating_regimes} defined the principal operating regimes supported by the four-thruster system. We now specialize to the encounter phase near a dominant gravitational primary, deriving the corresponding preferred configuration and authority limits.

\paragraph{In-plane authority: decrement-only steering.}

The three small thrusters provide in-plane steering authority. As established in Appendix~\ref{sec:thruster_theory}, symmetric warm firing ($u_1=u_2=u_3=u_0$) produces zero net force (Proposition~\ref{prop:warm_zero_force}), while steering is achieved through differential reallocation (Proposition~\ref{prop:differential_decomposition}). During encounter-phase stabilization, where avoiding ramp-up transients is favorable, the relevant near-primary specialization is the high-bias decrement-only regime (R2b).

In the high-bias warm state ($u_i \approx u_{\max}$), the small-jet triad admits a decrement-only steering mode: any in-plane force direction can be generated using only bounded throttle reductions relative to the warm baseline, with maximum achievable force magnitude set by the baseline level. This provides full latent in-plane authority while avoiding ramp-up transients in the most time-sensitive windows. The formal statement and proof are given in Proposition~\ref{prop:decrement_only}, together with an explicit extremal construction (Corollary~\ref{cor:lower_opposite_max}) in Appendix~\ref{sec:thruster_theory}.

\paragraph{Out-of-plane authority: sun-line alignment.}

In the near-primary encounter regime, the radial line toward the primary defines a dynamically distinguished axis. Let $\mathbf r(t)$ and $\mathbf v(t)$ denote position and velocity relative to the primary, with $\hat{\mathbf r}\equiv \mathbf r/\|\mathbf r\|$ and specific angular momentum $\mathbf h\equiv \mathbf r\times \mathbf v$. When the axial jet is aligned with the radial line, its actuation primarily changes orbital energy and timing while (in the two-body approximation) avoiding direct cross-track injection by preserving $\mathbf h$. Appendix~\ref{sec:axial_thrust_theory} collects the formal statements: radial thrust preserves $\mathbf h$ (Lemma~\ref{lem:radial_h}); gravity induces a natural bias structure yielding signed headroom in the radial channel (Proposition~\ref{prop:gravity_bias_authority}); and among fixed thrust directions this motivates radial (sun-line) alignment as the default during encounter-phase stabilization (Corollary~\ref{cor:sun_pointed}).

\paragraph{Unified stabilization configuration.}

Combining the in-plane and out-of-plane results yields a natural base configuration during encounter. A natural stabilization configuration during near-primary encounters is:
(i) operate the small-jet triad at high warm bias to retain full in-plane authority via decrement-only steering (Proposition~\ref{prop:decrement_only});
(ii) align the axial jet with the sun-line to concentrate high-thrust authority in the radial channel while avoiding direct cross-track injection (Corollary~\ref{cor:sun_pointed});
and (iii) treat the along-track direction as primarily momentum-dominated over short horizons.
The formal corollary statement is given in Appendix~\ref{sec:axial_thrust_theory} (Corollary~\ref{cor:base_config}).

In the encounter regime, the control channels therefore decompose naturally: in-plane authority from the small-jet triad in warm, high-bias operation; radial authority from the large jet structured by the gravitational field; and along-track stability dominated by existing momentum. This decomposition is induced by the encounter geometry and dynamics rather than imposed as a design constraint.

\section{Signature Characterization}
\label{sec:signatures}

Any propulsion system modifies the observable behavior of its host platform. This section characterizes photometric, astrometric, thermal, and spectroscopic signatures implied by the four-thruster architecture, treating the host comet as an observational baseline perturbed by controlled thrust and associated power and processing subsystems.

We interpret signatures using the operating-regime taxonomy defined in Section~\ref{sec:operating_regimes}; Figure~\ref{fig:signature_regimes} provides a schematic regime-to-pattern map used throughout this section.

\begin{figure}[h!]
    \centering
    \includegraphics[width=\textwidth]{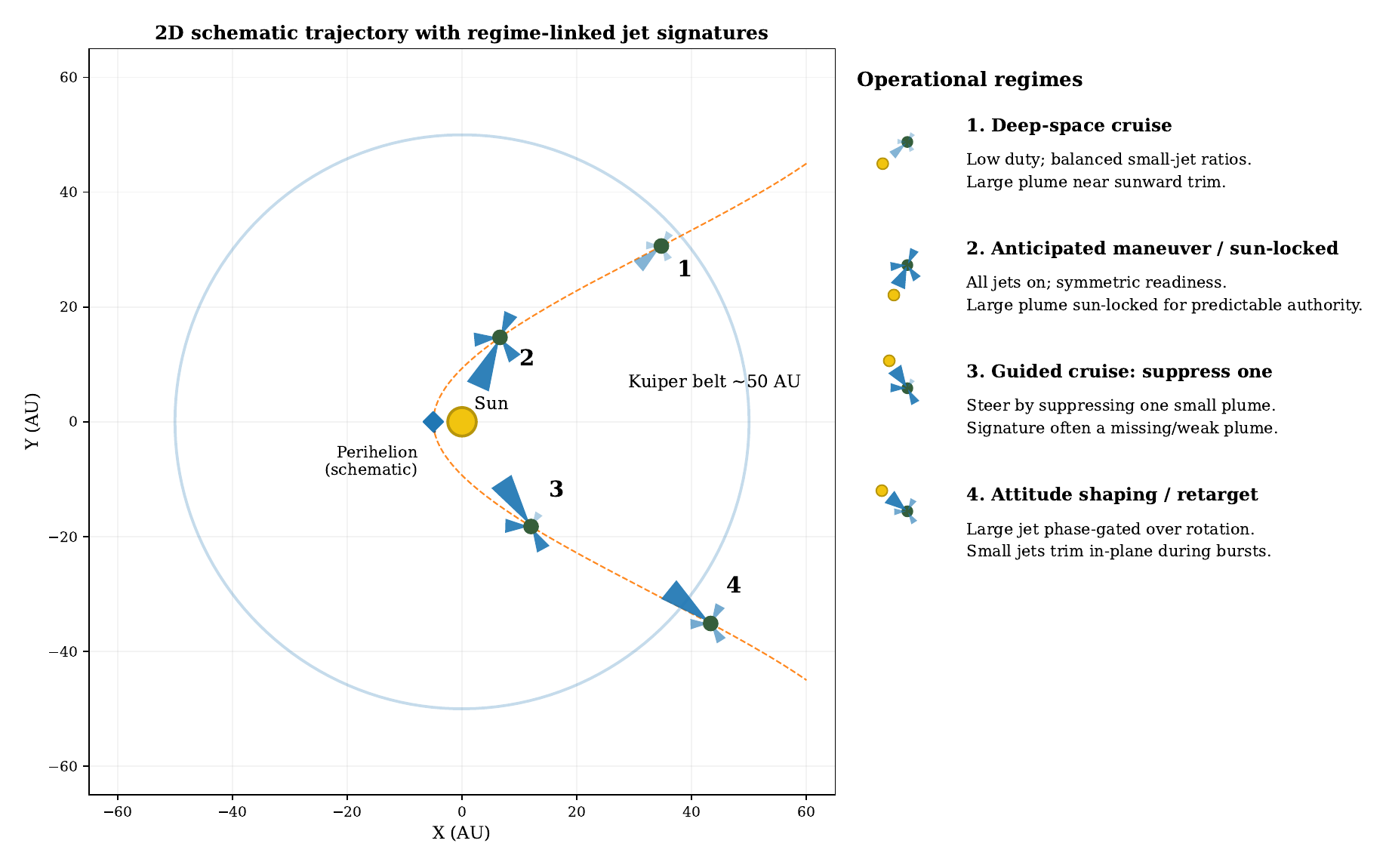}
    \caption{\textbf{Cone-limited navigation with regime-linked jet signatures (schematic).}
    Left: 2D heliocentric trajectory with four operating points (1--4); Kuiper belt radius ($\sim$50~AU) shown for scale. Right: corresponding jet-pattern icons for a $120^\circ$ three-jet system.
    \textbf{(1)} warm standby (R1), \textbf{(2)} elevated warm bias (R1/R2 readiness), \textbf{(3)} warm differential steering (R2), and \textbf{(4)} retargeting with axial-jet authority (R3) with possible thrust-neutral trim (R2a). Icon lengths/opacity indicate relative duty. The mapping from thrust allocations to net force is governed by Proposition~\ref{prop:warm_zero_force} and Proposition~\ref{prop:differential_decomposition}; fuel/brightness asymmetries under lowering versus raising follow Proposition~\ref{prop:lower_vs_raise}; and thrust-neutral reallocation follows Proposition~\ref{prop:thrust_neutral_unique}. Jet directions in the schematic are simplified as radially outward for visual clarity; the actual mounting geometry (tangential for T1--T3, axial for T4) is shown in Figure~\ref{fig:config} and discussed physically in Section~\ref{sec:nutation_architecture}.}
    \label{fig:signature_regimes}
\end{figure}

\subsection{Observable Signatures}
\label{sec:observable_signatures}

\paragraph{Photometric signatures.}

\textbf{Small thrusters (T1--T3).}
The three in-plane thrusters produce exhaust plumes whose brightness and relative intensities depend on regime:

\begin{itemize}
    \item \textbf{Warm standby (R1).} Three persistent emission features at $120^\circ$ spacing with comparatively stable integrated brightness. Net in-plane translation is exactly zero by Proposition~\ref{prop:warm_zero_force}, so the plume pattern is present without a corresponding transverse acceleration.

    \item \textbf{Warm differential steering (R2).} The brightness distribution becomes asymmetric as thrust is reallocated. A consequence of Proposition~\ref{prop:lower_vs_raise} is that, for a fixed steering magnitude, steering implemented by \emph{suppression} of one jet is more fuel-efficient than steering implemented by \emph{enhancement}. As a result, commanded acceleration can correlate more strongly with which plume is \emph{weaker} than with which plume is brighter.

    \item \textbf{Thrust-neutral reallocation (R2a).} By Proposition~\ref{prop:thrust_neutral_unique}, nonzero in-plane acceleration can be produced without changing total small-jet throughput. Photometrically, this would correspond to a redistribution of plume intensity at approximately constant integrated brightness, even while astrometric residuals accumulate (see below).

    \item \textbf{Decrement-only steering (R2b).} When operating about a high warm baseline, steering is achieved using only bounded reductions (Proposition~\ref{prop:decrement_only}). Photometrically, this would preferentially produce one or more \emph{dimming} plumes rather than a single strongly brightened plume during maneuvering.

    \item \textbf{Off / pulsed operation (R0).} Episodic brightenings corresponding to discrete burns, potentially involving only one or two jets (Proposition~\ref{prop:cold_start_two_thrusters}). Short-lived asymmetries during startup are consistent with Proposition~\ref{prop:ignition_transient}.
    
\end{itemize}

\textbf{Large thruster (T4).}
The axial thruster contributes an additional plume whose temporal structure is regime-dependent:
\begin{itemize}
    \item \textbf{Retargeting (R3).} Intermittent higher-throttle events can produce pronounced, phase-locked brightenings as the body rotates.
    \item \textbf{Bias operation (R4).} Continuous low-level firing yields a persistent brightness contribution that is periodic at the rotation frequency.
    \item \textbf{Off.} No contribution beyond natural outgassing.
\end{itemize}

Because the $120^\circ$ angular pattern is design-fixed while intensity ratios are control-driven, a potentially informative engineered signature would be a stable three-plume geometry whose \emph{relative} brightness varies with maneuvering and disturbance rejection, with weaker dependence on illumination than typical natural outgassing. It should be noted, however, that natural comets can also exhibit persistent jet structures from localized active regions, so geometric regularity alone would not constitute conclusive evidence of engineering.

\paragraph{Astrometric signatures.}

Thrust produces non-gravitational acceleration that appears as residuals in trajectory fits assuming purely gravitational dynamics. The gross character depends on duty cycle:

\begin{itemize}
    \item \textbf{Continuous operation (warm/bias regimes).} Secular residuals accumulating over the observational arc.
    \item \textbf{Pulsed operation (R0/R3).} Step-like changes in fitted elements associated with discrete impulses.
    \item \textbf{Near-primary radial alignment.} In the two-body approximation, radial thrust preserves the orbital plane (Lemma~\ref{lem:radial_h}); consequently, a bias-aligned channel tends to concentrate residuals in along-track timing/energy rather than cross-track drift, relative to transverse thrust.
\end{itemize}

A potentially informative case is thrust-neutral reallocation (R2a). Under this mode, significant astrometric residuals could occur \emph{without} commensurate changes in integrated small-jet brightness, because the maneuver redistributes plume intensity rather than increasing total throughput. Since total small-jet thrust remains constant under thrust-neutral steering, the photometric--astrometric decoupling is a direct, testable consequence of the $120^\circ$ geometry and would be difficult to reproduce with illumination-driven activity models.

In symmetric warm standby (R1), Proposition~\ref{prop:warm_zero_force} implies that any astrometric residual must arise from second-order effects (mismatch, misalignment, or plume--surface coupling) rather than commanded transverse acceleration.

\paragraph{Thermal signatures.}

Sustained operation at large heliocentric distances (beyond $\sim$10~AU) generally requires a non-solar primary power source in many architectures. Regardless of implementation (RTG or compact fission), the observational implication is a compact, body-fixed waste-heat emission region at radiator temperatures of order 300--500~K, approximately constant or slowly varying with heliocentric distance in contrast to solar-heated surfaces (which scale roughly as $r^{-1/2}$).

\paragraph{Spectroscopic signatures.}

Propellant extraction and processing alter the volatile inventory of the host body. Observable effects could include selective depletion of volatiles consumed as feedstock, abundance ratios inconsistent with typical cometary taxonomies or the body's inferred thermal history, and emission from intermediate species if fuel conversion is incomplete.

Natural comets, however, are known to exhibit compositional heterogeneity and seasonal variation in volatile release, so atypical ratios alone would not be sufficient to distinguish processing from natural variability without supporting evidence from other channels.

\subsection{Measurement and Discrimination Considerations}

\paragraph{Periodicity considerations.}

Detecting rotation-locked modulation requires high-cadence photometry---sampling every 1--2 hours for a body with a 7-hour rotation period---with precision of order 0.05~mag, sustained over multiple rotation cycles. These requirements are challenging but feasible for bright objects within $\sim$10~AU using current ground-based facilities.

A key discriminant is \emph{cross-channel consistency}: rotation-locked photometric modulation is most informative when it is contemporaneous with (or predictive of) epochs of anomalous astrometric residuals, strengthening the case that both arise from the same underlying thrust schedule rather than illumination-driven variability.

\paragraph{Cruise-phase duty cycle and signature suppression.}
For long-distance missions on multi-year to multi-decade trajectories, the nominal operating regime is expected to be \emph{cruise} rather than active maneuvering (Figure~\ref{fig:signature_regimes}, panel 1). Cruise phases occur predominantly at large heliocentric distances---beyond $\sim$10~AU for most of the transit---where both solar illumination and angular resolution from inner-system observers are greatly reduced. Efficient navigation further suppresses signatures: advance planning and smooth trajectory design minimize the frequency and amplitude of corrective maneuvers, so the propulsion system would typically operate in \emph{off}, \emph{idle}, or near-idle ``warm'' modes for extended periods, with higher-throttle firings occurring only intermittently for trajectory retargeting or disturbance rejection (panels 2--4).

The observational implications are compounding: intrinsically weak propulsion signatures are viewed at distances where detection is already challenging. Under nominal cruise operations, signatures are expected to be quasi-steady or confined to low-duty-cycle episodes. In symmetric warm idle, Proposition~\ref{prop:warm_zero_force} guarantees zero net force despite nonzero plume activity, so astrometric anomalies would be absent even though faint photometric signatures persist in principle. For a practical observer in the inner planetary system, the combination of large distance, low duty cycle, and low thrust levels makes cruise-phase detection exceedingly unlikely. Detectability would therefore be greatest during the brief encounter phase, when active maneuvering (differential steering per Propositions~\ref{prop:differential_decomposition}--\ref{prop:lower_vs_raise}) or sustained bias-thrust regimes produce stronger signatures at closer range.

\paragraph{Spatial resolution considerations.}

A jet plume extending $10^5$~km subtends approximately 14~arcsec at 10~AU, 3~arcsec at 50~AU, and 1.4~arcsec at 100~AU. Ground-based facilities (seeing-limited to $\sim$1~arcsec) can marginally resolve such structures within $\sim$50~AU; space-based observatories achieve better resolution. For unresolved sources, the signatures described above manifest as integrated photometric variations, analogous to asteroid light curve analysis.

\paragraph{Distinguishing signatures from natural activity.}
\label{sec:distinction_comparison}

Natural comets exhibit complex, time-variable activity driven by nucleus shape, spin state, and solar illumination. The propulsion-induced signatures described above would not be uniquely identifiable from any single observable. Instead, the most useful discriminants are those that (i) remain comparatively stable across changing illumination conditions, (ii) co-vary across independent measurement channels, and (iii) are difficult to reproduce with physically grounded outgassing models.

Table~\ref{tab:signatures} summarizes the potentially informative differences in compact form. Most entries---jet geometry, brightness variability, thermal emission, and compositional evolution---are discussed in the preceding subsections. Two channels merit additional comment because they involve cross-observable structure that the table cannot fully capture:

\begin{itemize}
    \item \textbf{Astrometric--photometric correlation}: Co-occurrence of non-gravitational accelerations with rotation-locked photometric changes is more constraining than either observable alone. The photometric--astrometric decoupling in Proposition~\ref{prop:thrust_neutral_unique} (nonzero acceleration without an integrated brightness change) would be particularly informative, as it is difficult to reproduce with standard outgassing models.

    \item \textbf{Flow-speed signatures}: Natural coma expansion velocities are typically $\sim0.5$--$1~\mathrm{km\,s^{-1}}$, depending on species and heliocentric distance~\cite{bockelee2004borrelly,jockers2011encke}. Conventional engineered thrusters expelling stored propellant generally achieve higher exhaust speeds---a few $\mathrm{km\,s^{-1}}$ for chemical propulsion and substantially higher for electric thrusters---potentially distinguishable via line widths, Doppler shifts, or collimation when spatially resolved~\cite{goebel2008thruster}. This contrast is conditional: an engineered system could instead direct or modulate native outgassing (e.g., via apertures or thermal control), in which case anisotropy, intermittency, or phase-angle invariance may be more diagnostic than absolute flow speed.
\end{itemize}

The evidentiary weight comes from \emph{co-occurrence across independent channels} and consistency with illumination-driven physical models, rather than from any single feature. Each individual signature can, in principle, be produced by natural processes; the discriminating power lies in their joint occurrence and internal consistency.

\begin{table}[h!]
\caption{Signature comparison: typical natural comet activity versus a comet-like body hosting sustained thrust.}
\label{tab:signatures}
\centering
\scriptsize
\setlength{\tabcolsep}{3pt}
\renewcommand{\arraystretch}{1.06}
\begin{tabular}{@{}p{0.22\textwidth}p{0.35\textwidth}p{0.35\textwidth}@{}}
\toprule
Observable & Typical natural activity & With sustained thrust \\
\midrule
Jet geometry &
Often irregular; may show persistent jets &
More persistent structure tied to thruster layout \\
Brightness variability &
Illumination-driven; can be rotation-modulated &
Phase-locked component less tied to illumination \\
Non-grav.\ acceleration &
Coupled to outgassing and illumination &
Residuals consistent with thrust duty cycle \\
Flow speed (if measurable) &
Typically $\sim$0.5--1~km\,s$^{-1}$ &
Generally higher (chemical: few km\,s$^{-1}$; electric: higher) \\
Thermal emission &
Solar-equilibrium trend ($\propto r^{-1/2}$) &
Localized radiator emission; weak $r$-dependence \\
Volatile ratios &
Taxonomic norms; heterogeneous &
Potential drift under sustained processing \\
\bottomrule
\end{tabular}
\end{table}

\paragraph{Forward modeling.}

The signatures described above can be quantitatively predicted through forward modeling. By jointly simulating thrust schedules, rotation phase, illumination geometry, and natural volatile release, one can generate synthetic photometric light curves, astrometric residuals, thermal evolution, and spectroscopic time series for direct comparison with observations. Key modeling parameters include thrust-to-mass ratio and duty cycle; idle-mode selection (off, symmetric warm per Proposition~\ref{prop:warm_zero_force}, or radial bias per Corollary~\ref{cor:sun_pointed}); steering strategy (differential per Propositions~\ref{prop:differential_decomposition}--\ref{prop:lower_vs_raise}, thrust-neutral per Proposition~\ref{prop:thrust_neutral_unique}, decrement-only per Proposition~\ref{prop:decrement_only}, or cold-start per Propositions~\ref{prop:cold_start_two_thrusters}--\ref{prop:cold_start_efficiency}); rotation period and spin-axis orientation; and natural outgassing rate and spatial distribution.

Beyond observability, forward models are essential for operational design. They determine which control laws are feasible under restart constraints (Proposition~\ref{prop:ignition_transient}) and plume--surface interactions, quantify the propellant cost of different warm-mode strategies, and map performance limits such as achievable curvature, response latency, and cross-track drift. In this way, modeling links engineering choices to mission envelopes---identifying which cruise-phase objectives (e.g., long-horizon retargeting, inclination changes, timing adjustments) are achievable for a given body and thruster set, and which phases would require supplemental strategies.

Because the four-thruster configuration is minimal under broad geometric and operational constraints, it represents a plausible convergent design point for long-horizon navigation using a naturally rotating, forward-moving body. The relevance of this observation is methodological rather than interpretive: minimal architectures imply structured, \emph{a priori} consequences that can be specified from engineering assumptions and tested against data via forward simulations. This also provides a practical basis for structured multiple-anomaly screening during characterization, as no single feature is diagnostic in isolation. The strongest constraint comes from cross-channel consistency and timing alignment---for example, whether photometric modulation, astrometric residual structure, thermal behavior, and compositional evolution co-vary in ways that are difficult to reproduce with illumination-driven activity alone---motivating escalation criteria for follow-up observations defined in advance rather than inferred post hoc~\cite{eldadi2025}.

We emphasize that the signature analysis in this section is intended to characterize the observational consequences of the proposed engineering architecture, not to establish a biosignature or technosignature detection framework. Whether any particular set of observational anomalies could be interpreted as evidence for artificial modification of a small body is a separate question requiring its own careful assessment, including detailed null-hypothesis modeling against the full range of known natural phenomena. Such an investigation is beyond the scope of the present work.

\subsection{Alternative Architecture: Nutation-Mediated Steering}
\label{sec:nutation_architecture}

The signature analysis above is developed under the principal-axis assumption
(Section~\ref{sec:frames_kinematics}), which reduces the attitude state to a single phase
variable and yields the clean separation between in-plane steering and out-of-plane
reorientation that underlies the four-thruster architecture. Relaxing this assumption leads
to an alternative design whose control structure and observable consequences differ in
instructive ways.

When non-principal-axis rotation is present or deliberately induced, the instantaneous spin
axis precesses around the angular momentum vector~$\mathbf{L}$, tracing a body cone with
nutation half-angle~$\theta$ governed by the ratio of rotational kinetic energy
to~$\|\mathbf{L}\|$. A body-fixed axial jet then sweeps through this cone in inertial space
over each precession cycle, gaining broader directional authority as~$\theta$ increases. This
suggests an inverted functional decomposition: the single large jet becomes the primary
translational steering actuator via nutation-mediated directional diversity, while the three
small jets at~$120^\circ$ intervals modulate spin dynamics---injecting rotational energy to
tighten the precession cone (reducing~$\theta$ for precision pointing) or extracting it to
widen the cone (increasing~$\theta$ for expanded steering authority). The minimum thruster
count is preserved: Theorem~\ref{thm:120} applies to torque synthesis under unidirectional
constraints just as it applies to force synthesis, so three small jets remain necessary for
full authority over spin-axis torques. More broadly, most force-synthesis results developed
in this paper---including differential decomposition
(Proposition~\ref{prop:differential_decomposition}), thrust-neutral steering
(Proposition~\ref{prop:thrust_neutral_unique}), and the operating-regime taxonomy of
Section~\ref{sec:operating_regimes}---are likely adaptable, since the underlying
conic-hull geometry is preserved when reinterpreted as torque allocation.

However, this architecture is not equivalent in complexity to the one developed above. As
noted in Remark~\ref{rem:dimensionality}, the principal-axis design reduces the effective
attitude state to a single phase variable, whereas the nutation-mediated design requires
active management of~$\theta$ as a controlled state coupled to the translational steering
envelope---elevating the attitude control problem from one to three effective dimensions.
The result is a hierarchically nested structure: the small jets control~$\theta$,
which sets the width of the steering cone, within which the large jet executes directional
thrust via phase selection over the precession cycle. This increased control-theoretic
complexity is the precise sense in which the principal-axis architecture is
the more minimal design: not in hardware count, which is identical, but in the
dimensionality of the coupled control problem that must be solved in real time.

The two architectures differ fundamentally in how thrust directions relate to
the body geometry. In the principal-axis design, no thruster is oriented
radially: the three small jets (T1--T3) fire tangentially in the equatorial
plane, producing direct translational steering forces, while the large jet
(T4) fires along the spin axis, generating torque through its off-axis moment
arm rather than direct translational authority. Out-of-plane steering arises
indirectly, as the moment arm sweeps the transverse torque plane over each
rotation cycle. Under nutation, these roles invert: because the spin axis
itself precesses around the angular momentum vector, T4's thrust direction
sweeps a cone in inertial space, converting it from a torque actuator into the
primary translational steering engine with directional authority widening as
the nutation angle~$\theta$ increases. The small jets correspondingly shift
from translational steering to spin-dynamics management, controlling~$\theta$
to set the width of the cone within which T4 can direct thrust. Intermediate
designs with tilted mounting would blend the two regimes, trading
instantaneous translational authority for enhanced torque coupling. This
functional inversion drives the signature differences described below.

Because the functional inversion primarily affects how thrust directions map
to inertial space, the observable signatures would shift in channels that
depend on attitude dynamics rather than on power or volatile processing.
Photometrically, the three small jets would produce intensity modulation at the
spin frequency rather than appearing as spatially fixed features, while the large jet would
generate a quasi-continuous plume whose apparent direction wobbles at the precession
frequency---potentially introducing two distinct periodicities rather than one. Symmetric
warm standby (Proposition~\ref{prop:warm_zero_force}) would reinterpret as zero net
\emph{torque} (maintaining constant~$\theta$) rather than zero net force, with
correspondingly different astrometric implications. Sun-line alignment (Corollary~\ref{cor:sun_pointed}) carries over as a constraint on the cone axis: the angular momentum vector $\mathbf{L}$ should be oriented along $\hat{\mathbf{r}}$, so that the instantaneous thrust direction sweeps a narrow cone centered on the sun-line. When the alignment is exact, the radial projection is uniform at $a_T\cos\theta$ around the precession cycle. In practice, residual misalignment between $\mathbf{L}$ and $\hat{\mathbf{r}}$ breaks this symmetry; throttle modulation synchronized to the precession phase---higher on the arc of closest radial alignment, lower on the opposite arc---can then concentrate the effective radial impulse while the cross-track components cancel on average over each cycle. Thermal, spectroscopic, and compositional signatures would
likely remain qualitatively unchanged, as these depend on power dissipation and volatile
processing rather than attitude dynamics. Future engineering and testing work should explore which body geometries, inertia ratios, and combinations of principal-axis and nutation-controlled architectures yield the most favorable flying devices for long-horizon missions. Observational evidence for non-principal-axis rotation in an interstellar comet, with a precession period consistent with this range, has been reported~\cite{scarmato2026wobble}.

\section{Discussion and Conclusion}
\label{sec:conclusion}

Motivated by interstellar comets as naturally translating, rotating platforms with volatile inventories, we study the minimal actuation needed to support long-horizon cruise-phase maneuvering and stabilization under unidirectional thrust constraints. Within the forward-cone reachability framework considered here, we show that three-dimensional steering for a rotating body with natural forward motion can be achieved with a minimal four-thruster architecture: three small units at $120^\circ$ intervals providing instantaneous in-plane steering, plus one larger unit that exploits rotation to supply low-bandwidth out-of-plane authority via attitude shaping. This separation between fast trim control and slow reorientation relaxes the need for instantaneous 6-DOF wrench controllability while retaining sustained navigational reachability over mission-relevant timescales. While developed for navigation, the same minimal-actuation architecture could also inform concepts for planetary defense by enabling sustained, low-complexity trajectory modification of hazardous comets or asteroids, complementing kinetic-impactor approaches demonstrated by DART~\cite{daly2023dart,thomas2023dart}.

Key results include: (i) a proof that the $120^\circ$ placement is the minimum-hardware solution for bidirectional planar thrust synthesis under unidirectional constraints (Theorem~\ref{thm:120}); (ii) bounds showing that geometric misalignment can substantially degrade worst-case authority; (iii) an explanation for why asymmetric thrust allocation (three small, one large) is structurally favored when frequent in-plane steering is separated from delayed out-of-plane reorientation (Corollary~\ref{cor:asymmetric}); and (iv) formal results establishing why sun-line alignment of the primary jet and high-bias warm modes for the small jets maximize stabilization authority while minimizing unintended trajectory perturbations (Corollary~\ref{cor:base_config}).

Trajectory-geometry constraints further motivate this division of labor. Debris-belt avoidance (implying $\sim$6$^\circ$ minimum slopes for Kuiper-belt-scale structures), node placement for inner-system reconnaissance, and propellant economy together yield a tightly constrained planning problem. Once a trajectory profile is selected, the dominant operational requirement shifts from episodic steering to continuous trajectory stabilization. In this regime, the large thruster (T4) naturally executes planned degree-scale slope adjustments at low frequency, while the three small thrusters (T1--T3) provide high-frequency in-plane stabilization. The minimal configuration thus reflects the structure of the cruise-phase problem: stability as the primary operational objective.

Given such an architecture and its operating modes, we then characterize the photometric, astrometric, spectroscopic, and thermal signatures that could follow. A central methodological point is that these observables are coupled outputs of a single underlying system rather than independent lines of evidence. Accordingly, diagnostic value depends on their \emph{joint} probability under physically grounded models, which is most naturally evaluated through probabilistic forward simulations rather than channel-by-channel interpretation. Given the short observational windows available for ISOs, simulation-based criteria may help identify observing conditions under which multilateral protocols for data sharing, joint analysis, and coordinated allocation of observing resources are most warranted~\cite{eldadi2025}.

Several limitations and extensions merit emphasis. More broadly, the geometric authority conditions established here---thruster placement ensuring full directional coverage, asymmetric allocation enabling separated control bandwidths---are necessary but not sufficient for active trajectory control stability. Sustained navigation additionally requires that the closed-loop tracking dynamics contract over time, so that stochastic disturbances do not accumulate into long-run trajectory divergence. The present paper establishes the spatial feedback conditions; temporal contraction under realistic forcing remains a separate analytical requirement. First, the present treatment idealizes time dependence: rotation is assumed steady-state and phase scheduling is evaluated over single cycles. Extending the reachability and authority results to stochastic perturbations---including outgassing torques, gravitational tides, and non-principal-axis rotation---would require a dynamical-systems analysis of the perturbed attitude dynamics, potentially including Lyapunov and robust-control tools. Second, the signature characterization here is intended to describe observational consequences of the proposed control system; it is not, by itself, a technosignature test. Attribution for any specific object requires dedicated null-hypothesis modeling and case-specific inference. Third, the principal-axis assumption that underpins the present architecture (Remark~\ref{rem:steady_spin_realism}) is a deliberate simplification: for natural cometary bodies, non-principal-axis rotation is arguably the more realistic baseline, and even an initially favorable spin state may not persist over mission-relevant timescales. The nutation-mediated architecture explored in Section~\ref{sec:nutation_architecture}---in which the thruster roles invert and the precession cone becomes a controlled steering parameter---therefore deserves further consideration as the physically more general setting. We have intentionally developed the principal-axis case first as a reduced-form starting point, both because it yields a lower-dimensional control problem amenable to closed-form analysis and because its results are conjectured to carry over, in adapted form, to the nutating case. A rigorous treatment of the nutation-controlled architecture, including stability analysis and steering-envelope characterization as a function of body inertia ratios, is left for future work.

Finally, practical performance will depend on device-level constraints, including plume--surface interactions, anchoring compliance, uncertain outgassing torques, sensor noise, and actuation delays. A natural next step is a suite of closed-loop simulations and design studies comparing feasible actuator suites and control laws, with forward models used to quantify achievable authority, propellant cost, robustness margins, and observability under realistic disturbances. Beyond interstellar exploration, the results may also inform small-body mission design where rotation-mediated control can reduce hardware requirements, including asteroid deflection and long-duration proximity operations.

\section*{Acknowledgments}

The findings, interpretations, and conclusions expressed herein are entirely those of the author and do not necessarily represent the views of the International Bank for Reconstruction and Development/World Bank, its Board of Executive Directors, or the governments they represent. No World Bank resources were used to conduct this research. This research received no specific funding. The author thanks Prof.\ Avi Loeb and acknowledges Dr.\ Frank Laukien for insights that informed this work.

\subsection*{Author Contributions}
B.P.J.\ Andr\'ee conceived the idea, developed the theoretical framework, performed all analyses, and wrote the manuscript.

\subsection*{Funding}
This research received no specific funding.

\subsection*{Conflicts of Interest}
The author declares that there is no conflict of interest regarding the publication of this article.

\section*{Supplementary Materials}

\noindent Accompanying work~\cite{andree2026stability} develops the stability theory and residual diagnostics for trajectory tracking under stochastic forcing.

\subsection*{Code Availability}
Reproducible code for the numerical illustrations is available as part of the supplementary materials. No experimental data were generated in this study.
Code and data to reproduce results and figures are archived on Zenodo at \url{https://doi.org/10.5281/zenodo.18986777}.

The supplementary file \texttt{analysis\_code.py} implements:

\begin{itemize}
    \item Verification of the $120^\circ$ coverage theorem
    \item Misalignment penalty calculations
    \item Control matrix computation and rank analysis
    \item Time-averaged torque space analysis
    \item Figure generation for all panels in Fig.~\ref{fig:config}
\end{itemize}

\subsection*{Additional Thruster Theory}
Supporting propositions and proofs are collected in Appendices~\ref{sec:thruster_theory} and~\ref{sec:axial_thrust_theory} below.

\appendix

\section{Three-Jet 120-Degree In-Plane System}
\label{sec:thruster_theory}

This appendix collects the formal results underlying the in-plane control authority of the three small thrusters. Given the number of geometric and operational claims used in the paper, the main text emphasizes synthesis and uses visual guidance (e.g., Figure~\ref{fig:signature_regimes}), while the appendix records the propositions, proofs, and auxiliary characterizations referenced throughout.

\subsection{Model and Notation}

We model three small thrusters producing in-plane forces along fixed unit directions
\begin{equation}
\mathbf d_1 = \begin{pmatrix}1\\0\end{pmatrix}, \qquad
\mathbf d_2 = \begin{pmatrix}-\tfrac{1}{2}\\ \tfrac{\sqrt{3}}{2}\end{pmatrix}, \qquad
\mathbf d_3 = \begin{pmatrix}-\tfrac{1}{2}\\ -\tfrac{\sqrt{3}}{2}\end{pmatrix},
\end{equation}
corresponding to tangential jets spaced at $120^\circ$ intervals. The net in-plane force is
\begin{equation}
\mathbf F = \sum_{i=1}^{3} u_i \mathbf d_i,
\end{equation}
where $u_i \ge 0$ denotes the thrust magnitude of thruster $i$ (no thrust reversal is possible).

\paragraph{Notation and conventions.}
Throughout this section:
\begin{itemize}
\item Indices $i, j, k$ range over $\{1, 2, 3\}$ unless otherwise stated.
\item Each $\mathbf d_i \in \mathbb{R}^2$ is a \emph{unit} vector: $\|\mathbf d_i\| = 1$.
\item For \emph{warm-bias} operation, we decompose $u_i = u_0 + \delta_i$, where $u_0 \ge 0$ is a 
      common baseline throttle and $\delta_i \in \mathbb{R}$ is the differential adjustment 
      (subject to the physical constraint $u_i \ge 0$).
\item We write $S \equiv \sum_{i=1}^{3} u_i$ for the total thrust. For fixed specific impulse $I_{sp}$, 
      propellant mass flow is proportional to $S$, so fuel-optimality corresponds to minimizing $S$ 
      for a given commanded force.
\item When distinguishing a commanded quantity from an incremental change, we use subscripts 
      (e.g., $u_{\mathrm{cmd}}$) rather than $\Delta$-prefixes for parameters, reserving 
      $\Delta(\cdot)$ exclusively for differences or errors (e.g., $\Delta S$, $\Delta u_i$).
\item Actuator saturation is modeled by bounds $0 \le u_i \le u_{\max}$, where $u_{\max} > 0$ 
      is the maximum thrust per thruster.
\end{itemize}

\paragraph{Time-varying notation (used in \S\ref{subsec:transients}).}
When analyzing ignition transients:
\begin{itemize}
\item $u_i(t)$ is the realized thrust of thruster $i$ at time $t$,
\item $u_{i,\mathrm{cmd}}(t)$ is the commanded thrust,
\item $\Delta u_i(t) \equiv u_i(t) - u_{i,\mathrm{cmd}}(t)$ is the actuation error,
\item $m > 0$ is the spacecraft mass (assumed constant over the maneuver),
\item $\mathbf J \equiv \int_0^{t_f} \mathbf F(t)\,dt$ denotes impulse over the interval $[0, t_f]$.
\end{itemize}

\subsection{Static Force Synthesis}

\begin{proposition}[Symmetric warm firing yields zero net translation]
\label{prop:warm_zero_force}
If all three thrusters fire at equal magnitude $u_0$ (symmetric warm mode),
\begin{equation}
u_1 = u_2 = u_3 = u_0,
\end{equation}
then the net in-plane force vanishes: $\mathbf F = \mathbf 0$.
\end{proposition}

\begin{proof}
Substitute the symmetric condition into the force equation and factor:
\[
\begin{aligned}
\mathbf F &= \sum_{i=1}^{3} u_0 \mathbf d_i \\
          &= u_0 \bigl(\mathbf d_1 + \mathbf d_2 + \mathbf d_3\bigr).
\end{aligned}
\]
Compute the vector sum directly using the $120^\circ$ geometry:
\[
\begin{aligned}
\mathbf d_1 + \mathbf d_2 + \mathbf d_3
&= \begin{pmatrix}1\\0\end{pmatrix}
+ \begin{pmatrix}-\tfrac{1}{2}\\ \tfrac{\sqrt{3}}{2}\end{pmatrix}
+ \begin{pmatrix}-\tfrac{1}{2}\\ -\tfrac{\sqrt{3}}{2}\end{pmatrix} \\[4pt]
&= \begin{pmatrix}1 - \tfrac{1}{2} - \tfrac{1}{2}\\ 0 + \tfrac{\sqrt{3}}{2} - \tfrac{\sqrt{3}}{2}\end{pmatrix} \\[4pt]
&= \begin{pmatrix}0\\0\end{pmatrix}.
\end{aligned}
\]
Hence $\mathbf F = u_0 \cdot \mathbf 0 = \mathbf 0$.
\end{proof}

\begin{proposition}[Differential steering decomposition]
\label{prop:differential_decomposition}
Let $u_i = u_0 + \delta_i$ with baseline $u_0 \ge 0$ and differentials $\delta_i \in \mathbb{R}$ 
satisfying $u_i \ge 0$ for all $i$. Then the net force and total thrust decompose as
\begin{equation}
\mathbf F = \sum_{i=1}^{3} \delta_i \mathbf d_i,
\qquad
S = 3u_0 + \sum_{i=1}^{3} \delta_i.
\end{equation}
In particular, the steering force depends only on the differentials $\delta_i$, and the change 
in total thrust relative to the symmetric baseline is $\Delta S = \sum_{i=1}^{3} \delta_i$.
\end{proposition}

\begin{proof}
Substitute the decomposition $u_i = u_0 + \delta_i$ into the force equation and distribute:
\[
\begin{aligned}
\mathbf F 
&= \sum_{i=1}^{3} (u_0 + \delta_i)\mathbf d_i \\
&= u_0 \sum_{i=1}^{3} \mathbf d_i + \sum_{i=1}^{3} \delta_i \mathbf d_i.
\end{aligned}
\]
By Proposition~\ref{prop:warm_zero_force}, the first sum vanishes, leaving 
$\mathbf F = \sum_{i=1}^{3} \delta_i \mathbf d_i$.

For the total thrust, sum over $i$:
\[
\begin{aligned}
S &= \sum_{i=1}^{3} u_i \\
  &= \sum_{i=1}^{3} (u_0 + \delta_i) \\
  &= 3u_0 + \sum_{i=1}^{3} \delta_i. \qedhere
\end{aligned}
\]
\end{proof}

\begin{proposition}[Same-axis steering with opposite fuel costs]
\label{prop:lower_vs_raise}
Starting from symmetric warm mode at baseline $u_0$, steering along the axis of thruster 
$k \in \{1,2,3\}$ can be achieved in two ways:
\[
\begin{aligned}
\textup{(raise)} &\quad \delta_k = +\epsilon, \quad \delta_i = 0 \ \text{for all } i \neq k, \\[2pt]
\textup{(lower)} &\quad \delta_k = -\epsilon, \quad \delta_i = 0 \ \text{for all } i \neq k,
\end{aligned}
\]
for some $\epsilon > 0$. These produce forces of equal magnitude but opposite direction and 
opposite fuel cost:
\[
\begin{aligned}
\textup{(raise)} &\quad \mathbf F = +\epsilon\,\mathbf d_k, \quad \Delta S = +\epsilon, \\[2pt]
\textup{(lower)} &\quad \mathbf F = -\epsilon\,\mathbf d_k, \quad \Delta S = -\epsilon.
\end{aligned}
\]
Thus, both maneuvers achieve identical steering authority $\|\mathbf F\| = \epsilon$ along the 
$\mathbf d_k$ axis, but raising increases propellant consumption while lowering decreases it 
by the same amount.
\end{proposition}

\begin{proof}
Apply Proposition~\ref{prop:differential_decomposition} with $\delta_k = \pm\epsilon$ and 
$\delta_i = 0$ for $i \neq k$. The steering force reduces to a single term:
\[
\mathbf F = \delta_k \mathbf d_k = (\pm\epsilon)\mathbf d_k.
\]
Hence:
\begin{itemize}
\item Raising ($\delta_k = +\epsilon$) yields $\mathbf F = +\epsilon\,\mathbf d_k$,
\item Lowering ($\delta_k = -\epsilon$) yields $\mathbf F = -\epsilon\,\mathbf d_k$.
\end{itemize}
In both cases $\|\mathbf F\| = \epsilon$ since $\|\mathbf d_k\| = 1$.

The total thrust increment is $\Delta S = \sum_{i=1}^{3} \delta_i = \delta_k$, giving 
$\Delta S_{\mathrm{raise}} = +\epsilon$ and $\Delta S_{\mathrm{lower}} = -\epsilon$.
\end{proof}

\begin{proposition}[Thrust-neutral steering: existence and uniqueness]
\label{prop:thrust_neutral_unique}
For any desired in-plane force $\mathbf F = (F_x, F_y)^\top$, there exists a \emph{unique} 
differential vector $\bm\delta = (\delta_1, \delta_2, \delta_3)^\top$ satisfying the 
\emph{thrust-neutral constraint}
\begin{equation}
\delta_1 + \delta_2 + \delta_3 = 0
\label{eq:thrust_neutral}
\end{equation}
and generating $\mathbf F$ via $\sum_{i=1}^{3} \delta_i \mathbf d_i = \mathbf F$. 
The unique solution is
\begin{align}
\delta_1 &= \tfrac{2}{3} F_x, \label{eq:delta1} \\
\delta_2 &= -\tfrac{1}{3} F_x + \tfrac{1}{\sqrt{3}} F_y, \label{eq:delta2} \\
\delta_3 &= -\tfrac{1}{3} F_x - \tfrac{1}{\sqrt{3}} F_y. \label{eq:delta3}
\end{align}
Consequently, thrust-neutral steering changes the net force without changing total thrust: 
$S$ remains equal to $3u_0$.
\end{proposition}

\begin{proof}
Write the force and neutrality conditions as a $3 \times 3$ linear system in 
$(\delta_1, \delta_2, \delta_3)$:
\begin{align}
F_x &= \delta_1 - \tfrac{1}{2}\delta_2 - \tfrac{1}{2}\delta_3, \label{eq:sys_x} \\
F_y &= \tfrac{\sqrt{3}}{2}\delta_2 - \tfrac{\sqrt{3}}{2}\delta_3, \label{eq:sys_y} \\
0   &= \delta_1 + \delta_2 + \delta_3. \label{eq:sys_neut}
\end{align}

\emph{Step 1: Solve for $\delta_1$.}
Rearrange~\eqref{eq:sys_neut} to obtain $\delta_2 + \delta_3 = -\delta_1$. 
Substitute into~\eqref{eq:sys_x}:
\[
\begin{aligned}
F_x &= \delta_1 - \tfrac{1}{2}(\delta_2 + \delta_3) \\
    &= \delta_1 - \tfrac{1}{2}(-\delta_1) \\
    &= \tfrac{3}{2}\delta_1.
\end{aligned}
\]
Solving yields $\delta_1 = \tfrac{2}{3} F_x$.

\emph{Step 2: Solve for $\delta_2$ and $\delta_3$.}
From~\eqref{eq:sys_y}, divide by $\tfrac{\sqrt{3}}{2}$ to get
\[
\delta_2 - \delta_3 = \tfrac{2}{\sqrt{3}} F_y.
\]
We now have two equations in $\delta_2, \delta_3$:
\[
\begin{aligned}
\delta_2 + \delta_3 &= -\tfrac{2}{3} F_x, \\
\delta_2 - \delta_3 &= \tfrac{2}{\sqrt{3}} F_y.
\end{aligned}
\]
Add these equations and divide by 2:
\[
\begin{aligned}
\delta_2 &= \tfrac{1}{2}\Bigl(-\tfrac{2}{3} F_x + \tfrac{2}{\sqrt{3}} F_y\Bigr) \\
         &= -\tfrac{1}{3} F_x + \tfrac{1}{\sqrt{3}} F_y.
\end{aligned}
\]
Subtract and divide by 2:
\[
\begin{aligned}
\delta_3 &= \tfrac{1}{2}\Bigl(-\tfrac{2}{3} F_x - \tfrac{2}{\sqrt{3}} F_y\Bigr) \\
         &= -\tfrac{1}{3} F_x - \tfrac{1}{\sqrt{3}} F_y.
\end{aligned}
\]

\emph{Uniqueness.}
The coefficient matrix of system~\eqref{eq:sys_x}--\eqref{eq:sys_neut} is
\[
A = \begin{pmatrix}
1 & -\tfrac{1}{2} & -\tfrac{1}{2} \\[2pt]
0 & \tfrac{\sqrt{3}}{2} & -\tfrac{\sqrt{3}}{2} \\[2pt]
1 & 1 & 1
\end{pmatrix}.
\]
A direct calculation gives $\det(A) = \tfrac{3\sqrt{3}}{2} \neq 0$, so the system has a unique 
solution for any $(F_x, F_y)$.

Finally, since $\sum_{i=1}^{3} \delta_i = 0$, Proposition~\ref{prop:differential_decomposition} 
implies $S = 3u_0 + 0 = 3u_0$.
\end{proof}

\begin{corollary}[Axis-aligned steering and fuel-flow constraints]
\label{cor:one_up_two_down}
The following results characterize thrust-neutral steering along thruster axes:

\begin{enumerate}
\item[\textup{(i)}] \textbf{One up, two down construction.}
To generate a thrust-neutral force collinear with $\mathbf d_k$ for any $k \in \{1,2,3\}$, 
set
\begin{equation}
\delta_k = +\epsilon, \qquad \delta_i = -\tfrac{\epsilon}{2} \ \text{for } i \neq k,
\end{equation}
yielding
\begin{equation}
\mathbf F = \tfrac{3\epsilon}{2}\,\mathbf d_k, \qquad \Delta S = 0.
\end{equation}

\item[\textup{(ii)}] \textbf{Fuel-flow constraints.}
Suppose total thrust is bounded below by $S \ge S_{\min}$ (e.g., thermal management or 
fuel-line pressure), with baseline $3u_0 = S_{\min}$. Then single-jet lowering 
(Proposition~\ref{prop:lower_vs_raise} with $\Delta S < 0$) is infeasible, but thrust-neutral 
steering via (i) remains fully available.
\end{enumerate}

Consequently, a connected three-jet system maintains full steering authority under fuel-flow 
constraints by \emph{rerouting} propellant between jets rather than reducing total flow.
\end{corollary}

\begin{proof}
\emph{Part (i).}
Impose thrust-neutrality $\delta_1 + \delta_2 + \delta_3 = 0$ together with the symmetry 
requirement $\delta_i = \delta_j$ for $i, j \neq k$ (necessary for the net force to lie 
along $\mathbf d_k$). Taking $k = 1$ without loss of generality, we have $\delta_2 = \delta_3$ 
and thus
\[
\delta_1 + 2\delta_2 = 0 \quad\Longrightarrow\quad \delta_2 = -\tfrac{\delta_1}{2}.
\]
Setting $\delta_1 = \epsilon$ gives $\delta_2 = \delta_3 = -\tfrac{\epsilon}{2}$.

The net force is
\[
\mathbf F = \epsilon\,\mathbf d_1 + \bigl(-\tfrac{\epsilon}{2}\bigr)(\mathbf d_2 + \mathbf d_3).
\]
From Proposition~\ref{prop:warm_zero_force}, $\mathbf d_1 + \mathbf d_2 + \mathbf d_3 = \mathbf 0$, 
so $\mathbf d_2 + \mathbf d_3 = -\mathbf d_1$. Substitute:
\[
\begin{aligned}
\mathbf F &= \epsilon\,\mathbf d_1 - \tfrac{\epsilon}{2}(-\mathbf d_1) \\
          &= \epsilon\,\mathbf d_1 + \tfrac{\epsilon}{2}\,\mathbf d_1 \\
          &= \tfrac{3\epsilon}{2}\,\mathbf d_1.
\end{aligned}
\]
Since $\sum_i \delta_i = \epsilon - \tfrac{\epsilon}{2} - \tfrac{\epsilon}{2} = 0$, we have 
$\Delta S = 0$.

\emph{Part (ii).}
If $S = 3u_0 = S_{\min}$ and we apply single-jet lowering with $\Delta S = -\epsilon < 0$, 
then $S + \Delta S < S_{\min}$, violating the constraint. However, thrust-neutral steering 
satisfies $\Delta S = 0$ by construction, so $S = S_{\min}$ is preserved. By 
Proposition~\ref{prop:thrust_neutral_unique}, any in-plane force $\mathbf F$ can be generated 
thrust-neutrally, and part (i) provides the explicit construction for axis-aligned forces.
\end{proof}

\subsection{Cold-Start Fuel Optimality}

\begin{proposition}[Cold-start optimum uses at most two thrusters]
\label{prop:cold_start_two_thrusters}
Consider the cold-start problem (no warm baseline): for a given nonzero target force 
$\mathbf F \in \mathbb{R}^2$, find thrust magnitudes solving
\begin{equation}
\begin{aligned}
\min_{u_1, u_2, u_3 \ge 0} \quad & u_1 + u_2 + u_3 \\
\text{subject to} \quad & u_1 \mathbf d_1 + u_2 \mathbf d_2 + u_3 \mathbf d_3 = \mathbf F.
\end{aligned}
\label{eq:cold_start_LP}
\end{equation}
Every optimal solution activates at most two thrusters (i.e., at least one $u_i = 0$).
\end{proposition}

\begin{proof}
\emph{Geometric intuition.}
The three thruster directions $\mathbf d_1, \mathbf d_2, \mathbf d_3$, spaced at $120^\circ$, 
partition the plane into three conic sectors. Any nonzero target force $\mathbf F$ lies in 
one of these sectors (or on a boundary between two). The two thrusters bounding that sector 
can generate any force within it via non-negative combinations---geometrically, $\mathbf F$ 
lies in their conic hull. The third thruster points into the opposite half-plane; activating 
it would require compensating thrust from the other two, strictly increasing total fuel 
consumption.

\emph{Formal argument.}
Problem~\eqref{eq:cold_start_LP} is a linear program in $\mathbb{R}^3$: the objective 
$\sum_i u_i$ and the constraint $\sum_i u_i \mathbf d_i = \mathbf F$ (two scalar equations) 
are linear, with non-negativity constraints $u_i \ge 0$.

By standard linear programming theory, an optimal solution occurs at an extreme point of 
the feasible polytope. In $\mathbb{R}^3$ with two equality constraints, extreme points 
require at least $3 - 2 = 1$ inequality constraint to be active, meaning at least one 
$u_i = 0$. Hence at most two thrusters are active at optimum.
\end{proof}

\begin{proposition}[Cold-start efficiency bounds]
\label{prop:cold_start_efficiency}
In cold-start mode with $120^\circ$ thruster placement, define the \emph{efficiency ratio} as 
the force magnitude achieved per unit thrust: $\eta = f / S^*$, where $f = \|\mathbf F\|$ and 
$S^*$ is the minimum total thrust from Proposition~\ref{prop:cold_start_two_thrusters}. Then
\begin{equation}
\frac{1}{2} \le \eta \le 1,
\end{equation}
or equivalently, $f \le S^* \le 2f$. Both bounds are tight:
\begin{itemize}
\item $\eta = 1$ (perfect efficiency) when $\mathbf F$ aligns with a thruster axis---one jet 
      suffices and all thrust contributes to the target direction.
\item $\eta = \tfrac{1}{2}$ (worst-case) when $\mathbf F$ points midway between two axes---both 
      jets fire equally, and their off-axis components cancel, wasting half the fuel.
\end{itemize}

\emph{Consequence.}
For fuel budgeting, cold-start maneuvers require at most twice the propellant of the ideal 
single-thruster case. This bounded inefficiency motivates warm-mode operation: thrust-neutral 
steering (Corollary~\ref{cor:one_up_two_down}) achieves arbitrary directions without the 
cancellation penalty, at the cost of maintaining a fuel-consuming baseline.
\end{proposition}

\begin{proof}
\emph{Lower bound on $S^*$ (upper bound on $\eta$).}
If $\mathbf F$ aligns with some $\mathbf d_k$, set $u_k = f$ and $u_i = 0$ for $i \neq k$. 
Then $\mathbf F = f\,\mathbf d_k$, $\|\mathbf F\| = f$, and $S^* = f$, giving $\eta = 1$.

\emph{Upper bound on $S^*$ (lower bound on $\eta$).}
Consider $\mathbf F$ pointing midway between $\mathbf d_1$ and $\mathbf d_2$, i.e., at angle 
$30^\circ$ from $\mathbf d_1$. By symmetry and Proposition~\ref{prop:cold_start_two_thrusters}, 
the optimum uses $u_1 = u_2 = u$ and $u_3 = 0$. The resulting force is
\[
\mathbf F = u(\mathbf d_1 + \mathbf d_2).
\]
Compute the magnitude of $\mathbf d_1 + \mathbf d_2$:
\[
\begin{aligned}
\mathbf d_1 + \mathbf d_2 
&= \begin{pmatrix}1 - \tfrac{1}{2}\\[2pt] 0 + \tfrac{\sqrt{3}}{2}\end{pmatrix}
= \begin{pmatrix}\tfrac{1}{2}\\[2pt] \tfrac{\sqrt{3}}{2}\end{pmatrix}, \\[6pt]
\|\mathbf d_1 + \mathbf d_2\| 
&= \sqrt{\tfrac{1}{4} + \tfrac{3}{4}} = 1.
\end{aligned}
\]
(This unit magnitude is a special property of $120^\circ$ spacing: pairwise sums of unit 
vectors at $120^\circ$ intervals are themselves unit vectors.)

Thus $f = \|\mathbf F\| = u$, while $S^* = 2u = 2f$, giving $\eta = \tfrac{1}{2}$.

\emph{No direction is worse.}
Any $\mathbf F$ lies within $60^\circ$ of some thruster axis (since axes are $120^\circ$ apart). 
The efficiency $\eta$ decreases monotonically from $1$ (on-axis) to $\tfrac{1}{2}$ (at the 
$60^\circ$ midpoint), confirming the bounds are tight and no direction yields $\eta < \tfrac{1}{2}$.
\end{proof}

\subsection{Actuator Limits and Steering Authority}

\begin{proposition}[Mid-throttle bias maximizes thrust-neutral steering authority]
\label{prop:headroom_mid}
Assume actuator bounds $0 \le u_i \le u_{\max}$ and warm-bias operation $u_i = u_0 + \delta_i$ 
with thrust-neutrality $\sum_{i=1}^{3} \delta_i = 0$. The feasible differential range for each 
thruster is
\[
-u_0 \le \delta_i \le u_{\max} - u_0.
\]
The \emph{symmetric headroom}---the largest magnitude $h$ such that $|\delta_i| \le h$ is 
feasible for all $i$---is
\[
h(u_0) = \min\{u_0,\, u_{\max} - u_0\}.
\]
This is maximized by choosing the baseline at mid-throttle:
\begin{equation}
u_0 = \frac{u_{\max}}{2}, \qquad h_{\max} = \frac{u_{\max}}{2}.
\end{equation}
\end{proposition}

\begin{proof}
The physical bounds $0 \le u_i \le u_{\max}$, combined with $u_i = u_0 + \delta_i$, imply
\[
0 \le u_0 + \delta_i \le u_{\max}
\quad\Longleftrightarrow\quad
-u_0 \le \delta_i \le u_{\max} - u_0.
\]
For symmetric authority in both directions (raising and lowering any thruster by up to $h$), 
we require $h \le u_0$ (room to lower) and $h \le u_{\max} - u_0$ (room to raise). Thus
\[
h(u_0) = \min\{u_0,\, u_{\max} - u_0\}.
\]
This piecewise-linear function of $u_0$ increases on $[0, u_{\max}/2]$ and decreases on 
$[u_{\max}/2, u_{\max}]$. The maximum occurs at the intersection $u_0 = u_{\max} - u_0$, 
i.e., $u_0 = u_{\max}/2$, yielding $h_{\max} = u_{\max}/2$.
\end{proof}

\subsection{High-Bias Operation: Decrement-Only Steering}

\begin{proposition}[Decrement-only steering about a high warm baseline]
\label{prop:decrement_only}
Assume bounds $0 \le u_i \le u_{\max}$. Consider a warm state with $u_i=u_0$ for all $i$, where
$u_0$ may be as large as $u_{\max}$. Write $u_i=u_0-\rho_i$ with $0 \le \rho_i \le u_0$. Then
\begin{equation}
\mathbf F = \sum_{i=1}^3 u_i\mathbf d_i = -\sum_{i=1}^{3} \rho_i \mathbf d_i.
\end{equation}
Any in-plane force direction is achievable using only reductions $\rho_i\ge 0$, and the maximum
achievable force magnitude satisfies $\|\mathbf F\|\le u_0$.
\end{proposition}

\begin{proof}
Using $\sum_{i=1}^3 \mathbf d_i=\mathbf 0$ (Proposition~\ref{prop:warm_zero_force}),
\[
\mathbf F=\sum_{i=1}^3 (u_0-\rho_i)\mathbf d_i
= u_0\sum_{i=1}^3 \mathbf d_i-\sum_{i=1}^3\rho_i\mathbf d_i
= -\sum_{i=1}^3\rho_i\mathbf d_i.
\]
Because the three unit directions at $120^\circ$ are not contained in any closed half-plane, the
conic hull of $\{\mathbf d_1,\mathbf d_2,\mathbf d_3\}$ equals $\mathbb R^2$; hence for any desired
direction $\hat{\mathbf f}\in\mathbb S^1$ there exists $\bm\rho\ge 0$ such that
$-\sum_i\rho_i\mathbf d_i$ is collinear with $\hat{\mathbf f}$.

For the magnitude, note that for any $\bm\rho\ge 0$,
\[
\|\mathbf F\|
=\Bigl\|\sum_{i=1}^3\rho_i\mathbf d_i\Bigr\|
\le \sum_{i=1}^3 \rho_i.
\]
Under decrement-only steering, at least one jet remains at the warm baseline in the extremal
construction (otherwise all three are reduced, which is never needed for maximizing a directed
force); therefore $\sum_i\rho_i\le u_0$ is sufficient to realize the maximum in any direction.
The bound $\|\mathbf F\|\le u_0$ is attained by the explicit constructions in
Corollary~\ref{cor:lower_opposite_max}.
\end{proof}

\begin{corollary}[``Lower the opposite side'' realizes maximal directional force]
\label{cor:lower_opposite_max}
Starting from $u_i=u_{\max}$, thrust along $\mathbf d_k$ is maximized by lowering the two opposite
jets equally (down to zero), while thrust along $-\mathbf d_k$ is maximized by lowering jet $k$
itself (down to zero). In both cases the maximal achievable magnitude equals $u_{\max}$.
\end{corollary}

\begin{proof}
Take $k=1$ without loss of generality.
To maximize thrust along $-\mathbf d_1$, set $u_1=0$ and $u_2=u_3=u_{\max}$, yielding
\[
\mathbf F=u_{\max}(\mathbf d_2+\mathbf d_3)=-u_{\max}\mathbf d_1
\]
since $\mathbf d_1+\mathbf d_2+\mathbf d_3=\mathbf 0$.

To maximize thrust along $+\mathbf d_1$, set $u_2=u_3=0$ and $u_1=u_{\max}$, yielding
$\mathbf F=u_{\max}\mathbf d_1$.

By symmetry the same argument holds for any $k\in\{1,2,3\}$.
\end{proof}

\subsection{Ignition Transients}
\label{subsec:transients}

\begin{proposition}[Ignition imbalance induces unintended impulse]
\label{prop:ignition_transient}
Consider a commanded multi-thruster ignition intended to produce time-varying force 
$\mathbf F_{\mathrm{cmd}}(t) = \sum_{i=1}^{3} u_{i,\mathrm{cmd}}(t)\,\mathbf d_i$.
Let the realized thrust be $u_i(t) = u_{i,\mathrm{cmd}}(t) + \Delta u_i(t)$, where 
$\Delta u_i(t)$ captures startup delay, non-response, or throttle mismatch. Then:

\begin{enumerate}
\item[\textup{(i)}] The unintended force is
\begin{equation}
\mathbf F_{\mathrm{err}}(t) = \sum_{i=1}^{3} \Delta u_i(t)\,\mathbf d_i.
\end{equation}

\item[\textup{(ii)}] The resulting velocity error over the interval $[0, t_f]$ satisfies
\begin{equation}
\|\Delta \mathbf v_{\mathrm{err}}\|
\le \frac{1}{m} \int_0^{t_f} \sum_{i=1}^{3} |\Delta u_i(t)|\,dt.
\label{eq:velocity_error_bound}
\end{equation}

\item[\textup{(iii)}] If a single thruster $k \in \{1,2,3\}$ fails to ignite for duration $\tau$ 
while commanded at constant level $u_{\mathrm{cmd}} > 0$ (with other thrusters nominal), then
\begin{equation}
\|\mathbf J_{\mathrm{err}}\| = u_{\mathrm{cmd}}\,\tau,
\qquad
\|\Delta \mathbf v_{\mathrm{err}}\| = \frac{u_{\mathrm{cmd}}\,\tau}{m}.
\label{eq:single_thruster_failure}
\end{equation}
\end{enumerate}
Thus ignition imbalance is most consequential at high throttle and short timescales, motivating 
warm-mode or ramped ignition strategies when transient deviations are costly.
\end{proposition}

\begin{proof}
\emph{Part (i).}
Substitute the error decomposition into the force equation and distribute:
\begin{align}
\mathbf F(t)
&= \sum_{i=1}^{3} u_i(t)\,\mathbf d_i \notag \\
&= \sum_{i=1}^{3} \bigl(u_{i,\mathrm{cmd}}(t) + \Delta u_i(t)\bigr)\mathbf d_i \notag \\
&= \sum_{i=1}^{3} u_{i,\mathrm{cmd}}(t)\,\mathbf d_i 
   + \sum_{i=1}^{3} \Delta u_i(t)\,\mathbf d_i \notag \\
&= \mathbf F_{\mathrm{cmd}}(t) + \mathbf F_{\mathrm{err}}(t).
\end{align}

\emph{Part (ii).}
The velocity error is the integral of force error divided by mass:
\[
\Delta \mathbf v_{\mathrm{err}} 
= \frac{1}{m} \int_0^{t_f} \mathbf F_{\mathrm{err}}(t)\,dt.
\]
Apply the norm-integral inequality $\bigl\|\int \mathbf g\,dt\bigr\| \le \int \|\mathbf g\|\,dt$:
\[
\begin{aligned}
\|\Delta \mathbf v_{\mathrm{err}}\|
&\le \frac{1}{m} \int_0^{t_f} \|\mathbf F_{\mathrm{err}}(t)\|\,dt \\
&= \frac{1}{m} \int_0^{t_f} \Bigl\|\sum_{i=1}^{3} \Delta u_i(t)\,\mathbf d_i\Bigr\|\,dt.
\end{aligned}
\]
Apply the triangle inequality $\|\sum_i \mathbf a_i\| \le \sum_i \|\mathbf a_i\|$ and use 
$\|\mathbf d_i\| = 1$:
\[
\begin{aligned}
\|\Delta \mathbf v_{\mathrm{err}}\|
&\le \frac{1}{m} \int_0^{t_f} \sum_{i=1}^{3} \|\Delta u_i(t)\,\mathbf d_i\|\,dt \\
&= \frac{1}{m} \int_0^{t_f} \sum_{i=1}^{3} |\Delta u_i(t)|\,dt.
\end{aligned}
\]

\emph{Part (iii).}
For the single-thruster failure scenario: thruster $k$ is commanded at $u_{\mathrm{cmd}}$ but 
produces zero thrust for duration $\tau$, so $\Delta u_k(t) = 0 - u_{\mathrm{cmd}} = -u_{\mathrm{cmd}}$ 
on $[0, \tau]$. All other thrusters perform nominally: $\Delta u_i(t) = 0$ for $i \neq k$. Thus
\[
\begin{aligned}
\mathbf F_{\mathrm{err}}(t) &= -u_{\mathrm{cmd}}\,\mathbf d_k, \\
\|\mathbf F_{\mathrm{err}}(t)\| &= u_{\mathrm{cmd}}.
\end{aligned}
\]
The impulse error over $[0, \tau]$ is
\[
\begin{aligned}
\mathbf J_{\mathrm{err}} 
&= \int_0^\tau \mathbf F_{\mathrm{err}}(t)\,dt \\
&= -u_{\mathrm{cmd}}\,\tau\,\mathbf d_k,
\end{aligned}
\]
so $\|\mathbf J_{\mathrm{err}}\| = u_{\mathrm{cmd}}\,\tau$. Dividing by mass gives 
$\|\Delta \mathbf v_{\mathrm{err}}\| = u_{\mathrm{cmd}}\,\tau / m$.
\end{proof}

\begin{remark}[Operational risk mitigation]
\label{rem:risk_mitigation}
Proposition~\ref{prop:ignition_transient} quantifies why cold-start ignition is risky: a 
single thruster failing to ignite at high throttle immediately produces a large unintended 
impulse. Several mitigation strategies follow:

\begin{enumerate}
\item[\textup{(i)}] \textbf{Warm-mode operation.}
Maintaining a symmetric baseline $u_0$ allows steering via differentials, with ignition 
transients occurring only when transitioning from complete shutdown.

\item[\textup{(ii)}] \textbf{Ramped ignition.}
Gradually increasing $u_{\mathrm{cmd}}(t)$ from zero bounds $|\Delta u_i(t)|$ during the 
critical startup phase, reducing the magnitude of any imbalance-induced impulse.

\item[\textup{(iii)}] \textbf{High-thrust jet orientation.}
Since transient risk scales with $u_{\mathrm{cmd}}$, the highest-thrust actuator dominates ignition and shutdown hazards. In encounter phases near a dominant primary, aligning the axial jet with the radial line has two distinct benefits. First, thrust (and thrust error) that is predominantly radial preserves the orbital plane in the two-body approximation, so large short-lived transients couple mainly into energy and timing rather than cross-track drift (Lemma~\ref{lem:radial_h}). Second, the gravity-biased structure of the radial channel provides a natural operating point: when outward thrust is used to partially offset inward gravity, the net radial acceleration lies in the signed interval of Proposition~\ref{prop:gravity_bias_authority}. In particular, operating near the bias $a_T\approx g(r)$ makes the \emph{net} radial acceleration small in magnitude, so ignition shortfall or throttle mismatch perturbs the net radial acceleration around that bias rather than introducing a large transverse impulse. This is the formal motivation for sun-line (radial) alignment in stabilization modes (Corollary~\ref{cor:sun_pointed}).
\end{enumerate}
\end{remark}

\section{Axial Jet Authority Near a Dominant Primary}
\label{sec:axial_thrust_theory}

This appendix collects the formal results governing the out-of-plane (axial) control channel in the encounter regime near a dominant gravitational primary. The main text applies these results in Section~\ref{sec:operational_config} and in the signature discussion (Section~\ref{sec:signatures}).

Notation is as follows. Let $\mathbf r(t)$ and $\mathbf v(t)$ denote position and velocity relative to the primary, with $\hat{\mathbf r}\equiv \mathbf r/\|\mathbf r\|$. Let $\mu$ be the gravitational parameter and define $g(r)=\mu/r^2$. The specific angular momentum is $\mathbf h\equiv \mathbf r\times \mathbf v$. We consider an axial thruster producing acceleration $\mathbf a_T(t)$; the radial (sun-line) specialization takes $\mathbf a_T(t)=a_T(t)\hat{\mathbf r}$.

\begin{lemma}[Radial thrust preserves orbital angular momentum]
\label{lem:radial_h}
Under central gravity with additional thrust acceleration $\mathbf a_T(t)=a_T(t)\hat{\mathbf r}$ (purely radial),
\begin{equation}
\dot{\mathbf h}=\mathbf 0.
\end{equation}
Radial thrust modifies orbital energy and timing but does not rotate the orbital plane (no cross-track injection).
\end{lemma}

\begin{proof}
$\dot{\mathbf h}=\frac{d}{dt}(\mathbf r\times \mathbf v)=\mathbf r\times \dot{\mathbf v}$. Central gravity is radial, hence $\mathbf r\times \mathbf a_g=\mathbf 0$. Under purely radial thrust, $\mathbf r\times \mathbf a_T=\mathbf 0$ as well. Therefore $\dot{\mathbf h}=\mathbf 0$.
\end{proof}

Before specializing to the radial case, we consider which fixed thrust direction best serves encounter-phase stabilization. A mission planner operating near a dominant primary would prefer a configuration in which the dominant forces are predictable and perturbations are minimized. Pointing along-track opposing motion acts as a brake, bleeding orbital energy that is expensive to recover during a high-speed transit. Pointing toward the primary adds thrust to gravity, deepening the gravitational well and creating a net inward acceleration whose loss of control authority near perihelion could be unrecoverable. Pointing along-track with motion continuously accelerates during the flyby, altering the encounter geometry throughout a maneuver that was designed for a specific speed profile. Pointing in a time-varying or arbitrary direction injects energy into all orbital elements simultaneously, and a mission planner seeking predictability would avoid coupling perturbations into channels that are otherwise quiescent. The radial (sun-line) direction is distinguished because the three dominant forces acting on the body along this axis---gravity (inward), outgassing (predominantly outward for thermally driven sublimation), and thrust (outward)---are collinear, with two partially canceling before thrust is applied. The residual is small and its magnitude is governed by parameters the mission planner controls, as the next result makes precise.

\begin{proposition}[Radial force balance near a dominant primary]
\label{prop:gravity_bias_authority}
Let $g(r)=\mu/r^2$ denote gravitational acceleration toward the primary and let $a_{\mathrm{out}}(r)\ge 0$ denote the outgassing-induced acceleration along $+\hat{\mathbf r}$, predominantly anti-sunward for thermally driven sublimation. If the axial jet produces controllable acceleration $a_T\in[a_{\min},a_{\max}]$ along $+\hat{\mathbf r}$ (outward), the net radial acceleration is
\begin{equation}
a_{\mathrm{net}} \;=\; a_T + a_{\mathrm{out}}(r) - g(r) + \varepsilon_t,
\label{eq:radial_balance}
\end{equation}
where $\varepsilon_t$ captures stochastic perturbations (irregular outgassing, radiation-pressure transients, micrometeorite impulses) with $\mathbb{E}[\varepsilon_t]=0$. The heliocentric distance $r$ is a trajectory-design parameter that jointly determines $g(r)$ and $a_{\mathrm{out}}(r)$, and $a_T$ is set by throttle within $[a_{\min},a_{\max}]$. Selecting
\begin{equation}
a_T^* \;=\; g(r)-a_{\mathrm{out}}(r)
\label{eq:aT_star}
\end{equation}
(feasible whenever $a_{\min}\le g(r)-a_{\mathrm{out}}(r)\le a_{\max}$) zeros out all deterministic terms, leaving
\begin{equation}
a_{\mathrm{net}} \;=\; \varepsilon_t,
\end{equation}
so that the operational problem reduces to perturbation management. The signed headroom about this equilibrium is
\begin{equation}
h \;=\; \min\!\bigl\{a_T^*-a_{\min},\;\; a_{\max}-a_T^*\bigr\},
\label{eq:headroom}
\end{equation}
representing the control authority available to reject realizations of $\varepsilon_t$. Maintaining $a_T^*>a_{\min}+\delta$ for a safety margin $\delta>0$ ensures that perturbations up to magnitude $\delta$ can be absorbed without saturating the actuator. The orbital plane is preserved throughout by Lemma~\ref{lem:radial_h}.
\end{proposition}

\begin{proof}
The net radial acceleration is the superposition of outward thrust, outward outgassing, inward gravity, and stochastic forcing: $a_{\mathrm{net}}=a_T+a_{\mathrm{out}}(r)-g(r)+\varepsilon_t$. Substituting $a_T^*=g(r)-a_{\mathrm{out}}(r)$ gives $a_{\mathrm{net}}=\varepsilon_t$. Since $a_T$ can deviate from $a_T^*$ by at most $a_T^*-a_{\min}$ downward and $a_{\max}-a_T^*$ upward, the signed headroom follows. Plane preservation holds because thrust remains radial (Lemma~\ref{lem:radial_h}).
\end{proof}

\begin{corollary}[Sun-line alignment in stabilization modes]
\label{cor:sun_pointed}
Among fixed thrust directions, sun-line alignment uniquely satisfies three properties established above: (i) all deterministic radial forces are collinear and can be zeroed by choice of $r$ and $a_T$ (Proposition~\ref{prop:gravity_bias_authority}); (ii) the orbital plane is preserved (Lemma~\ref{lem:radial_h}); and (iii) the residual operational problem is rejection of zero-mean stochastic perturbations within the available headroom $h$.
\end{corollary}

\begin{corollary}[Encounter-phase base configuration for stabilization]
\label{cor:base_config}
In the encounter regime near a dominant primary, a natural stabilization configuration is:
\begin{enumerate}
\item[\textup{(i)}] \textbf{Small jets (T1--T3):} operate at high warm bias to retain full in-plane
authority via decrement-only steering (Proposition~\ref{prop:decrement_only}).
\item[\textup{(ii)}] \textbf{Large jet (T4):} align with the sun-line to concentrate high-thrust
authority in the radial channel while preserving $\mathbf h$ (Lemma~\ref{lem:radial_h}) and exploiting
gravity-biased radial headroom (Proposition~\ref{prop:gravity_bias_authority}).
\item[\textup{(iii)}] \textbf{Along-track:} over short encounter windows, along-track evolution is
primarily momentum-dominated and is typically managed through timing/energy adjustments in the radial
channel rather than direct along-track thrust allocation.
\end{enumerate}
\end{corollary}

A clarification on the sign convention used in Corollary~\ref{cor:base_config} is warranted.

\begin{remark}[Nozzle convention]
``Sun-line alignment'' refers to the \emph{axis} set by the radial direction. In the analysis above, $\hat{\mathbf r}$ is outward from the primary and $a_T\hat{\mathbf r}$ denotes an \emph{outward} (anti-sunward) applied acceleration that counteracts inward gravity. If phrasing is given in terms of where the nozzle points, then ``point the nozzle at the Sun'' corresponds to an outward thrust force.
\end{remark}

\printbibliography

\end{document}